\documentclass[12pt,twoside]{article}

\usepackage[margin=1.3in]{geometry}
\usepackage{enumerate}

\usepackage{amsmath,amssymb,euscript}
\usepackage{amssymb,mathrsfs}

\usepackage{amsmath}
\usepackage{amsthm}
\usepackage{amssymb}
\usepackage{graphicx}

\usepackage{makeidx}

\usepackage{hyperref}

\newtheorem{theorem}{Theorem}[section]

\newtheorem{lemma}[theorem]{Lemma}

\newtheorem{proposition}[theorem]{Proposition}
\newtheorem{corollary}[theorem]{Corollary}

\newtheorem{remark}[theorem]{Remark}

\numberwithin{theorem}{section} \numberwithin{equation}{section}

\newcommand{\nc}{\newcommand}

\nc{\be}{\begin{equation}}
\nc{\la}{\label}
\nc{\ba}{\begin{array}}
\nc{\ea}{\end{array}}
\nc{\bs}{\begin{split}}
\nc{\es}{\end{split}}

\newcommand{\R}{{\mathbb R}}
\newcommand{\Z}{{\mathbb Z}}

\newcommand{\C}{{\mathbb C}}

\newcommand{\bH}{{\mathbb H}}

\newcommand{\A}{{\bf A}}
\newcommand{\J}{{\bf J}}
\newcommand{\bfa}{{\bf a}}
\renewcommand{\a}{{\bf a}}

\newcommand{\cE}{\mathcal{E}}

\newcommand{\cH}{\mathcal{H}}

\newcommand{\cL}{\mathcal{L}}         

\nc{\eps}{\epsilon}
\nc{\e}{\epsilon}
\nc{\lam}{\lambda}
\nc{\G}{\Gamma}
\nc{\g}{\gamma}
\nc{\al}{\alpha}
\nc{\del}{\delta}
\nc{\om}{\omega}
\nc{\Om}{\Omega}
\nc{\ta}{\tau}
\nc{\w}{\omega}
\nc{\io}{\iota}
\nc{\h}{\theta}
\nc{\z}{\zeta}
\nc{\s}{\sigma}
\nc{\Si}{\Sigma}
\nc{\Lam}{\Lambda}

\nc{\bP}{\bar{P}}
\nc{\bQ}{\bar{Q}}
\nc{\bL}{\bar{L}}

\nc{\ra}{\rightarrow}

\nc{\ran}{\rangle}
\nc{\lan}{\langle}
\newcommand{\Null}{\operatorname{Null}}

\newcommand{\Ran}[1]{\operatorname{Ran}\, #1}
\newcommand{\re}{\operatorname{Re}}
\newcommand{\one}{{\bf 1}}

\newcommand{\im}{\operatorname{Im}}

\newcommand{\sign}{\operatorname{sign}}

\newcommand{\ls}{\lesssim}

\newcommand{\n}{\nabla}
\newcommand{\p}{\partial}

\newcommand{\divv}{\operatorname{div}}
\newcommand{\curl}{\operatorname{curl}}

\nc{\na}{\nabla_a}
\nc{\dA}{\nabla_A}

\newcommand{\DETAILS}[1]{}

\newcommand{\variants}[1]{}

\newcommand{\avg}[1]{\left\langle #1 \right\rangle} 

\newcommand{\Abs}[1]{| #1 |}

\newcounter{tol}
\def\sp(#1,#2){\left\langle #1 \, , #2 \right\rangle}
\newcommand{\bp}[1]{\setcounter{tol}{0}
	\foreach \j in {#1} {\stepcounter{tol}} 
	\left\langle%
	\foreach \i [count=\ni] in {#1} {
		\ifnum\ni=\arabic{tol}
		\i%
		\else%
		\i,%
		\fi%
	}
	\right\rangle%
}

\newcommand{\Lpsi}[2]{\mathcal{L}_{#2}^{2}} 
\newcommand{\LA}[2]{\vec{\mathcal{L}}_{}^{2}}
\newcommand{\HA}[2]{\vec{\mathcal{H}}_{}^{2}} 

\newcommand{\an}{a^n}
\newcommand{\psio}{\psi_0}
\newcommand{\phio}{\psi_0}

\newcommand{\lat}{\mathcal{L}}

\newcommand{\latn}{\hat{\mathcal L}} 

\nc{\Omt}{\tilde{\Omega}}
\nc{\Oml}{\Omega^\lat}
\nc{\Omn}{\hat{\Omega}} 

\begin{document}
\title{Vortex lattice solutions of the ZHK Chern-Simons equations}
 \date{October 19, 2019} 
\author{K. Rajaratnam$^*$ 
 and  I.M. Sigal\footnote{Dept. of Math., U. of Toronto, Toronto, 
 Canada, krishan2500@gmail.com,  imsigal@gmail.com} 
 }
\maketitle

\begin{abstract} We consider the non-relativistic Chern-Simons equations proposed by Zhang, Hansen and Kivelson as the mean field theory of the fractional Hall effect.
We prove the existence of the vortex lattice solutions (i.e. solution with lattice symmetry and with topological degree one per lattice cell) similar to the Abrikosov solutions of superconductivity. We derive an asymptotic expression for the energy per unit area and show that it attains minimum at the hexagonal lattice.\end{abstract}

\section{Introduction}

The Chern-Simons topological invariant defines a gauge theory in odd space-time dimensions.  
Coupled to matter, this theory leads to a number of closely related models, which play an important role in Particle and Condensed Matter Physics (see \cite{Girv, {KivLeeZhang}, Dun, HorvZh, Zhang} for some reviews). 

In this paper we study  Chern-Simons (CS) gauge theories distinguished by (a) the presence of external magnetic fields and (b) self-interaction potentials of the double well 
 type. 
In the non-relativistic case, such a theory is 
  due to Zhang, Hansen and Kivelson (the ZHK model, \cite{ZHK, Zhang}). It gives a mean-field description to the fractional quantum Hall effect (FQHE). This theory is defined in terms of the action which, in the dimensionless form, 
  is given by \begin{align}\label{NR-CS-action'}
S(\Psi, A, A_0):=\int i \bar \Psi \p_{t A_0} \Psi - 
 |\n_{ A} \Psi |^2 - V(|\Psi|^2) +  \A' \cdot\curl \A'. 
\end{align}
where  $\Psi:\R_+\times \R^2\ra \C$ is the order parameter of composite bosons, $A =A' + A^{\rm ext} $ and $A_0=A_0' + A^{\rm ext}_0$ are the total vector and scalar potentials from $\R_+\times \R^2$ to $\R^2$ and $\R$, respectively,  $ \A' = (A', A'_0)$ and $ \A^{\rm ext} =(A^{\rm ext}, A^{\rm ext}_0)$ are the CS gauge field\footnote{For the origin of the CS action in the FQHE, see \cite{FrEtAl}} and external electro-magnetic potential, and, taking the electron charge to be $-1$,  $\n_{ A}:=\n +i  A$ and $ \p_{t A_0}:=\p_{t} - i A_0$.

Furthermore, 
 $V$ is a self-interaction energy (nonlinearity) given by the double well potential 
\begin{align} \label{V-dw}V(|\Psi|^2) = \frac{g}{2}(|\Psi|^2 - \mu)^2,\end{align}
where $g, \mu >0$. (A related model,  
was suggested by Jackiw and Pi. There, one takes $b = \curl A^{ext} = 0$, and \eqref{V-dw} with $\mu = 0$.\footnote{For other non-relativistic Chern-Simons theories, see \cite{BH1, BH2, Man}.} It follows from results of \cite{BH1, BH2} that having a constant magnetic field is equivalent to  including a background charge density, see \cite{BH1, BH2, Dun}.) 

The last term in \eqref{NR-CS-action} gives the Chern-Simons topological invariant. Indeed, note that $\A' \cdot\curl \A'=\al' \wedge d \al'$, where  $\al'$ is the one-form  $\al'= A_0' d x^0 +A_1' d x^1 +  A_2' d x^2$.
 
Let $v(s):=V'(s)$. Then the Euler-Lagrange equations for \eqref{NR-CS-action} state 
\DETAILS{We consider the non-relativistic Chern-Simons gauge theories described by the order parameter $\phi$,  and vector and scalar potentials $\vec a$ and $a_0$. Given  external potentials $\vec a^{\rm ext}$ and $a_0^{\rm ext}$, $(\phi, \vec a, a_0)$ satisfy the evolution equations} 
\begin{subequations} \label{CSeqs}\begin{align}\label{NR-CS-eq1}
& i\p_{t } \Psi= (-  \Delta_{ A} - A_0)\Psi + v(|\Psi|^2)\Psi,\\  
\label{NR-CS-eq2} &   *\p_t  A = \im (\bar \Psi \n_{ A} \Psi)-  \curl^* A_0,\\
\label{NR-CS-eq3} & 0=  \curl  A -  b + \frac12 |\Psi|^2, \end{align}\end{subequations}
with 
$-\Delta_{ A}:=\n_{ A}^*\n_{ A}$,  where the operator $\nabla_{A}^*$  is the $L^2$ adjoint of $\nabla_{A}$ (it is  given by $\nabla_{A}^* F = - \divv F + i A \cdot F$), 
 $* (v_1, v_2)=(-v_2, v_1)$ (so that $\curl=- \divv *$ and $\curl^*=- * \n$), $A_0\equiv A_0'$
  and $b=\curl A^{\rm ext} $. 
 (Since the electron charge is negative, the electric potential $A_0$ enters with the minus sign.)

We consider  \eqref{CSeqs} on the local Sobolev space, $H^2_{\rm loc}$, of order $2$. Note that
\begin{itemize}
	\item $B = \curl A$ and $E=\nabla A_0 - \partial_t A$ are the total magnetic and electric fields \footnote{$\nabla A_0$ enters with an unusual sign, since $A_0$ is the emlectric potential times the electron charge.};
	\item  $\rho = |\Psi|^2$  and $J = \im(\bar{\Psi} \nabla_A \Psi)$ are the charge and current densities of `composite bosons' (see e.g. \cite{ReadRez, Tong}). 
 \end{itemize}

 \eqref{CSeqs} will be called the {\it Zhang-Hanson-Kivelson (ZHK) equations}.  
We concentrate on the ZHK model, 
but with a more general nonlinearity, $V(s)$, namely, one satisfying 
\begin{align} \label{V-cond}V\in C^2,\  V(s) \geq 0\  \quad \forall s \geq 0,\  \quad V'(0) < 0,\ V''(0) > 0.\end{align}

As in the standard ZHK model, we also assume that the external magnetic field is a constant  and that there is no external electric field, 
\begin{align} \label{Aext-cond} 
 \curl A^{ext} = b >0\ \text{  and }\  A^{ext}_0 = 0.\end{align} 
(The case $b<0$ is considered similarly.) 
 In this case, Eqs.  \eqref{CSeqs} are translational invariant. 

We are interested in the {\it ground states} of \eqref{CSeqs}, i.e.  static solutions with the lowest energy locally. 
 First we examine the most symmetric solutions. Since  Eqs.  \eqref{CSeqs} are invariant under gauge transformations and translations (see Appendix \ref{sec:CSmodels}), our first candidate for the ground state is a homogeneous, or gauge-translation invariant\footnote{
Because of gauge invariance, a symmetry is defined up to gauge transformations.}
 solution, satisfying \eqref{Aext-cond}: 
 \begin{align}\label{norm-sol-orig} v_*^b:= (\Psi=0,\  
   A= A^{\rm ext}=A^{b},  A_0=0), 
  \end{align}  
 where $A^{b}$ denotes a vector potential 
  with $\curl A^{b}= b$. 
   Since $\Psi=0$, we call $v_*^b$ the {\it normal} solution.
  
Let $\chi:= -v(0)=-V'(0)>0$.  
One can show (see \cite{S2}) that the gauge-translationally invariant (normal) solution $v_*^b$ is linearly stable
  for  $b> \chi$ and unstable for $b<  \chi$. 
   (See e.g. \cite{S, S2} for the definition of stability.)

  In this paper, we show that 
  at  $b=\chi$ new solutions with lower energy per unit area emerge from the normal one. These solutions break the translational symmetry to a lattice one\footnote{More precisely, they are invariant, up to gauge transformations under lattice translations.} and have topological degree one per lattice cell. 
We also show that their energy per unit area attains minimum at the hexagonal lattice.

We call the solutions mentioned above the {\it vortex lattice states}. They are analogous to solutions in superconductivity discovered by A. A. Abrikosov, who 
 defined them as states whose all associated physical quantities are periodic w.r.to a lattice in $\R^2$. 
One can show (see e.g. \cite{S}) that $(\Psi, A, A_0)$ is an Abrikosov lattice state, for some  lattice $\lat$, iff it satisfies 
  \begin{align}            \label{gaugeperiod-gen}  &\Psi(x + s) = e^{i \chi_s(x)}\Psi(x),\  A(x + s) =  A(x)+ \n \chi_s(x),\ A_0(x + s) = A_0(x), 
  \end{align}
for some differentiable $\chi_s(x)\in \R$ and any $s  \in \lat$. One can show also that $(\Phi, A, A_0)$ satisfying \eqref{gaugeperiod-gen} is gauge equivalent to one satisfying \eqref{gaugeperiod-gen}  with 
   \begin{align} \label{chis-spec} 
   \chi_s(x):=\frac{ b'}{2}x\cdot J s+ c_s,\ c_s \text{ satisfies }\ c_{s+t} - c_s - c_t - \frac{1}{2} b' s \wedge t \in 2\pi\Z, \end{align}
where $ b' = \frac{2\pi n}{|\lat|}$, with $n\in \Z$ and $|\lat|$, the area of any fundamental (elementary) cell of $\lat$ (all elementary cells have the same area). Note that (a) the vector potential $A^{b'}$ satisfies \eqref{gaugeperiod-gen}-\eqref{chis-spec}  
and (b) for any $A$ satisfying the middle relation in \eqref{gaugeperiod-gen} with \eqref{chis-spec}, 
   the magnetic flux through an arbitrary fundamental cell, $\Om^{\lat}$, of $\lat$ quantized as 
\begin{align}  \label{flux-quant}\frac{1}{2\pi} \int_{\Om^{\lat}} \curl  A = n\in \Z.\end{align}

We will call   $v:=(\Psi, A, A_0)\in C(\R^2, \C\times \R^2\times\R)$ satisfying \eqref{gaugeperiod-gen}-\eqref{chis-spec}  an $\lat$-{\it equivariant} state, or, because of \eqref{flux-quant}, a {\it vortex lattice}. To avoid cumbersome expressions we will use these terms only {\it for $n=1$}.

Recall the notation $\chi:= -V'(0)\equiv - v(0)>0$ and let $g:= V''(0)\equiv  v'(0)>0$.  Denote by $\hat\lat$ lattices with fundamental cells of area $2\pi$ ({\it normalized} lattices).  Our first result is the following 
  \begin{theorem}\label{thm:ALexist} 
Assume \eqref{V-cond} and \eqref{Aext-cond}. 
 Then for any 
  normalized lattice $\hat\lat$, any $\chi >0$, and any $b > 0$, s.t. $0 < (\chi-b)/(g-1)\ll 1$, 
   there is an $r=r^{b, \hat\lat}>0$ s.t.
    
(a)  \eqref{CSeqs} has a non-trivial, $r\hat\lat$-equivariant,  static  $H^2_{\rm loc}-$solution, $v^{b, \hat\lat}$, unique among  $r\hat\lat$-equivariant $H^2_{\rm loc}-$triples in a vicinity of \eqref{norm-sol-orig}; 

(b)	this solution bifurcates from the normal branch of solutions, \eqref{norm-sol-orig}, at  $b= \chi$;

(c) 
 for $g>1$, 
the energy of these solutions per unit area is smaller than the energy of the homogeneous state \eqref{norm-sol-orig}.   \end{theorem}
This theorem is proven in Sections \ref{sec:bifurc-dimK=1} and \ref{sec:b-s-relat}. More precise properties of the solutions $v^{\hat\lat, b}$ are described in Theorem \ref{thm:bif-thm-Kdim1} and Corollary \ref{cor:ub}.

The time-translational and gauge invariance of equations \eqref{CSeqs} imply that the energy, \begin{align}\label{en'}
E_{ \Oml}^{b}(\Psi, A, A_0):=\int_{ \Oml}  |\n_{ A} \Psi |^2 + V(|\Psi|^2) - 2 A_0( \frac{1}{2}|\Psi|^2 -  b +   \curl  A) 
\end{align}
and the charge, $\int_{ \Oml} |\Psi|^2$,
are conserved and for $(\Psi, A, A_0)$ an $\lat$-{\it equivariant} solution to \eqref{CSeqs}. (In fact, the same is true for any bounded domain $Q\subset \R^2$ with appropriate boundary conditions, Noether's theorem).  Note also that \eqref{en'} can be obtained in the canonical way from action \eqref{NR-CS-action'}.
\DETAILS{on a set constrained by \eqref{NR-CS-eq3'}, with 
 $A_0$ being a Lagrange multiplier.}
 We see now that the r.h.s. of equations \eqref{CSeqs} is the $L^2-$gradient of this functional w.r.to $(\Psi, A, A_0)$.

Our second result deals with the energetics of the bifurcating solutions 
for various lattices. 
\begin{theorem}[Lattice energy]  \label{thm:eng.min} 	
	Assume \eqref{V-cond} and \eqref{Aext-cond} and let $0 < \chi-b\ll g-1$. 
	  Then, as $b\ra \chi$, the solutions, $v^{b, \hat\lat}$, found in the previous theorem, with the lowest energy per lattice cell, correspond to the 
	 lattices approaching a 
	  hexagonal lattice. (The topology on the set of the normalized lattices is given in terms of the complex parametrization  described at the beginning of Section \ref{sec:b-s-relat}.)	
	\end{theorem}

For our next result, we need more definitions. We define  the Abrikosov function 
\begin{equation}\label{beta'}
	\beta(\lat) \equiv \beta(\hat\lat) = \frac{\avg{|\psio|^4}}{{\avg{|\psio|^2}}^2}
\end{equation}
where $\psio$ is the unique solution to the equation 
$\p_{A}\psio = 0$, where $\p_{A}  := \frac12 ((\nabla_{A})_1 - i(\nabla_{A })_2)$ is the complexified covariant derivative (see Appendix \ref{sec:self-dual} and e.g. \cite{S, S2}), satisfying  the first relation in \eqref{gaugeperiod-gen} with \eqref{chis-spec} and  $ b' = \frac{2\pi }{|\lat|}$, 
 and $\avg{f}:=\frac{1}{|\Om^{\lat}|} \int_{\Om^{\lat}} f$ ($|\Om^{\lat}|\equiv |\lat|$).  

Next, let  $ \hat\Om$ be a fundamental cell of $ \hat\lat$. For the solutions $v^{b, \hat\lat}\equiv (\Psi^{b}, A^{b}, A_0^{b})$, found in Theorem \ref{thm:ALexist}, we define  the energy (see \eqref{en'})
\begin{align}  \label{Eb-lat} 	& E_b(\hat\lat)= (r^{b, \hat\lat})^{-1}E_{r^{b, \hat\lat} \hat\Om}^{b}(v^{b, \hat\lat}).\end{align}
 The last statement of Theorem \ref{thm:ALexist} and Theorem \ref{thm:eng.min} will follow from the following result on the energy asymptotic.
\begin{theorem}[Energy asymptotic] \label{thm:eng-asymp}
	Assume \eqref{V-cond}, \eqref{Aext-cond} and  $(\chi-b)( g-1)>0$. 
	Let $\mu:=(\chi-b)/( g-1)$. Then for  $|\chi-b|\ll |g-1|$, 
\DETAILS	{\begin{align}\label{energy-resc''}
&E_b(\tau)  =V(0) -\frac{1}{2 } \frac{(1-g)\beta(\tau)}{ [g\beta(\tau)]^{2}}\mu^2 
+O(\mu^3) ,
\end{align}}	
\begin{align}  \label{Eb-lat-asymp} 
	& E_b(\hat\lat)	= V(0) - \frac12 \frac{g-1}{\beta(\hat\lat)}\mu^2 
  + O(\mu^3).\end{align}
\end{theorem}
 
\begin{proof}[Proof of 
Theorem \ref{thm:ALexist}(c)] 
 Expression \eqref{Eb-lat-asymp} implies Theorem \ref{thm:ALexist}(c). 
  \end{proof}

  Theorems \ref{thm:ALexist},  \ref{thm:eng.min}  and  \ref{thm:eng-asymp} give the first result on  existence of  vortex lattices for the  ZHK equations and on their local energies. 
   
\paragraph{Related results.}
The existence of 
 vortex lattices for the self-dual, non-relativistic JP 
 model (i.e. for a single, albeit important, value of  the coupling constant  $g$) 
 was obtained in \cite{ACSvH, Ol}. 
  (The existence of topological vortices 
was shown in \cite{GuoLi}.)

\paragraph{Remarks.} 1)   We show in Appendix \ref{sec:self-dual} that  for self-interaction potential \eqref{V-dw},
 the ZHK equations for $\lat$-equivariant fields are in the self-dual  regime iff $g=1$.

2)  
Theorem \ref{thm:eng-asymp} shows that the local energy is (counterintuitively) raised if $g<1$ and lowered if $g>1$. This, suggests that there are two types of materials here similarly to type I and type II superconductors. 
 
3) For \eqref{V-dw} and $b =\mu$, there is the additional homogeneous solution  
  \[\Psi=\sqrt \mu,\ A =0, 
  \] 
which corresponds to the total  condensation. In this solution, the CS gauge field, $A'$, opposes and cancels the external field, $A^{\rm ext}$. 
For \eqref{V-dw}, with $1/g = 2\pi (2k+1), k\in \Z$, these are 
 the uniform Laughlin states of feeling fractions $\nu=1/(2k-1)$.  

4) \eqref{flux-quant} is a special case of the Chern-Weil correspondence for compact Riemann surfaces (see e.g. \cite{CERS} and the references therein). 

 5) The magnetic flux quantization, \eqref{flux-quant}, and  Eq \eqref{NR-CS-eq3} 
   imply the charge quantization 
\begin{align}\label{rho-quant}
\frac{1}{2}\int_{ \Oml} \rho - b=\int_{ \Oml}  \curl A'  \in 2\pi  \Z.
\end{align}

6) Clearly, $v^b_\om  := (\Psi=0, A_0=\om, A=A^{b}) = (0, \om, A^{b})$, for any $\om\in \R$, is also a gauge-translationally invariant solution. Moreover, for different $\om$, $v^b_\om$ are gauge inequivalent. From physics viewpoint, $v^b_\om$ is the standard normal solution  for the nonlinearity $V'(|\Psi|^2)+\om\ (=g |\Psi|^2-\chi+\om)$ instead of $V'(|\Psi|^2)\ (=g |\Psi|^2-\chi)$, i.e.  for a different chemical potential. 

7) Under constraint \eqref{NR-CS-eq3}, 
energy functional \eqref{en'} becomes

  \begin{align}\label{en}
E_Q (\Psi,  A):=\int_Q  |\n_{ A} \Psi |^2 + V(|\Psi|^2) , 
\end{align}

8) For the symmetries, hamiltonian structure and the conservation laws see Appendix \ref{sec:CSmodels}.
 
9) 
The stability/instability of the 
  Abrikosov lattice solutions is investigated in \cite{Raj}.

\bigskip

Our paper is organized as follows. After some preliminary discussions in  Section \ref{sec:prelim},  we prove Theorems \ref{thm:ALexist}(a, b) 
 and \ref{thm:eng-asymp} in 
 Sections \ref{sec:modif-eqs} - \ref{sec:b-s-relat} and \ref{sec:energy-asymp}, respectively. 
In Section \ref{sec:eng.min}, we prove 
 Theorem \ref{thm:eng.min}. 
 Some background on and properties of the Chern-Simons gauge theories are given in Appendix \ref{sec:CSmodels}.  
 For convenience of the reader we present in Appendix \ref{sec:current-repr} a proof of a result from a related paper. 


\subsection*{Acknowledgements}
The authors thank  Ilias Chenn and Dmitri Chouchkov for many stimulating discussions. 
 The research on this paper is supported in part by NSERC Grant No. NA7901.

\section{Rescaling, gauge fixing, reformulation}\label{sec:prelim} 
In this section, we lay out the background for the proof of Theorem \ref{thm:ALexist}.  
We are interested in static solutions of \eqref{CSeqs},  i.e.  solutions for which $\Psi$, $ A$ and $ A_0$ are time independent. Hence $(\Psi, A, A_0)$ solve the equations 
\begin{subequations} \label{CSeqs-stat}
\begin{align}\label{NR-CS-eq1-stat}
  - & \Delta_{ A} \Psi + 
v(|\Psi|^2)  \Psi - A_0 \Psi=0,\\  
\label{NR-CS-eq2-stat} & \im (\bar \Psi \n_{ A} \Psi)-  \curl^*  A_0=0,\\ 
\label{NR-CS-eq3-stat}& \curl A+\frac12 |\Psi|^2-  b 
 =0, 
\end{align}\end{subequations} 
where, recall,  $v(s):=V'(s)$. 

To be more specific, we look for $(\Psi, A, A_0)$ equivariant w.r.to a lattice $\lat$, which we consider as one of the unknowns. In particular, we will vary the area, $|\lat|$, of a fundamental cell 
 of $\lat$.

In order to eliminate the unknown $|\lat|$ 
 from the definition of the spaces we rescale the solutions to the fixed, 
 normalized lattice
   \begin{equation}\label{LATnorm}\latn:=\lam^{-1/2} \lat,\ \qquad \lam:= \frac{| \lat|}{2\pi }.  
     \end{equation} 
  We note that $|\hat\lat| = 2\pi $. 
We define the rescaled fields $(\psi,  a, a_0)$ as 
\begin{align}\label{resc}
    (\psi(x), a(x), a_0(x)) :=& ( \sqrt \lam \Psi(\sqrt\lam x), \sqrt\lam  A(\sqrt\lam x), \lam A_0(\sqrt\lam x)). 
\end{align}

We summarize the effects of this rescaling: 
    \begin{enumerate}
    \item[(A)] $(\Psi, A, A_0)$ solves the static ZHK equations, \eqref{CSeqs-stat}, 
    if and only if $(\psi, a, a_0)$ solves
\begin{subequations} \label{CSeqs-stat-resc'} \begin{align}\label{CSeqs-stat-resc'1}
  - &  \Delta_{ a} \psi + 
v_\lam (|\psi|^2)   \psi- a_0 \psi=0,\\  
\label{CSeqs-stat-resc'2} & \im (\bar \psi \n_{ a} \psi)-  \curl^*  a_0=0,\\ 
\label{CSeqs-stat-resc'3}&  \curl a- \lam b+\frac12 |\psi|^2 =0, 
\end{align} \end{subequations} 
where  $v_\lam (|\psi|^2):=\lam V'(\lam^{-1}|\psi|^2)\equiv \lam v(\lam^{-1}|\psi|^2)$.

    \item[(B)]   \label{reduced-gauge-form}
  $(\Psi, A, A_0)$ is   $\lat$-equivariant  iff $(\psi, a, a_0)$ is  $\latn$-equivariant, i.e. it satisfies 
\begin{subequations} \label{gaugeper}   \begin{align} \label{psi-gaugeper'} 
 &\psi(x + s) = e^{i(\frac{ n}{2}x\cdot J s+ c_s) }\psi(x),\\  
 \label{a-per} &  a(x + s) = a(x)+ \frac{ n}{2} J s,\\  
   \label{a0-per'} & a_0(x + s) = a_0(x), 
  \end{align}\end{subequations}  
for every $s  \in \latn$, where  $c_s$ satisfies the condition 
  \begin{align} \label{cs-cond} c_{s+t} - c_s - c_t - \frac{1}{2} n s \wedge t \in 2\pi\Z. \end{align}  
 The magnetic flux quantization, \eqref{flux-quant}, for the rescaled $a$ states 	
\begin{align}  \label{flux-quant-resc}\frac{1}{2\pi} \int_{\hat\Om} \curl  a = n\in \Z,\end{align}
for any elementary cell,  $\Omn$, of the lattice $\latn$.
 	
	\item[(C)]   \label{rescal-en} The rescaled average energy per lattice cell, obtained from \eqref{en}  is 
	\begin{align} \label{en-resc'} 
		\cE_\lam (\psi, a) & = \frac{1}{|{\lat}|} E_{\Omega^{\lat}}( \frac{1}{\sqrt{\lam}}\psi(\frac{1}{\sqrt{\lam}} x), \frac{1}{\sqrt{\lam}} a(\frac{1}{\sqrt{\lam}} x)) \\
 \label{en-resc}		& = \frac{1}{2 \pi}\int_{\Omn} (\frac{1}{ \lam^2}|\nabla_{ a} \psi|^2 + V(\frac{1}{\lam}|\psi|^2)) d x 
  \end{align}    
where   $\Om^{\lat}$ denotes an elementary cell of the lattice $\lat$, provided \eqref{CSeqs-stat-resc'3} holds.  \end{enumerate}

Let $ a^{n}$ denote the vector potential with the constant magnetic field $\curl a^n=n$. For any $b$, equations \eqref{CSeqs-stat-resc'} 
 for $
 (\psi,  a, a_0, \lam)$, have the normal (`trivial') solution 
 \begin{align}\label{norm-sol} 
(\psi=0,\   a= a^{n},\ a_0=0,\ \lam= \frac{n}{b}). 
 \end{align}

Recall that we consider $| \lat|$ and therefore $\lam:= \frac{2\pi }{|\lat|}$ as unknowns. To have the same number of equations as unknowns, $(\psi, a, a_0, \lam)$, we take the cell-average of \eqref{CSeqs-stat-resc'3} and use \eqref{flux-quant-resc} 
  to obtain the fourth equation. 
  After setting  $\theta:=\lam-\frac nb$  and
  \begin{align} \label{a'-deco}  a= a^{n}+  \al,\end{align} and passing to the unknowns $u:=(\psi,  \al, a_0, \theta)$, the system \eqref{CSeqs-stat-resc'} becomes 
\begin{subequations} \label{CSeqs-stat-resc}  
\begin{align}\label{CSeqs-stat-resc1}
  - &  \Delta_{ a} \psi + 
v_\lam (|\psi|^2)   \psi-  
a_0 \psi=0,\\  
\label{CSeqs-stat-resc2} &  \curl^*a_0- \im (\bar \psi \n_{a} \psi)=0,\\  
\label{CSeqs-stat-resc3} &  \curl \al + \frac12  |\psi|^2  
-\frac12 \lan |\psi|^2\ran=0,\\ 
\label{CSeqs-stat-resc4}&
- b \theta +\frac12 \lan |\psi|^2\ran =0,\\
\label{a-lam-deco}& a=a^{n}+  \al,\ \lam = \frac nb+\theta.
\end{align}\end{subequations} 
Here and in what follows, we denote 
$\lan f\ran:= \frac{1}{|\hat\Om|}\int_{\hat\Om} f$.
 Furthermore, assume that $\psi$, $ \al$ and $a_0$ satisfy, for every $s  \in \latn$, 
 \begin{subequations} \label{gaugeper'}    \begin{align}\label{psi-gaugeper}  &\psi(x + s) = e^{i(\frac{ n}{2}x\cdot J s+c_s) }\psi(x),\\  
   \label{alpha-per} & \alpha(x + s) =  \alpha(x)\ 
  {\rm and}\ \divv  \alpha = 0,\\ 
 \label{a0-per} &a_0(x + s) = a_0(x). \end{align} \end{subequations}


\section{Modified equations}\label{sec:modif-eqs} 
 In this section, adopting the strategy from \cite{CERS} (see also \cite{TS1, CSS}), we replace the original system of equations, \eqref{CSeqs-stat-resc}, by one giving the same solutions and which is easier to handle. 
 
 We begin with defining appropriate spaces we work in.
 Let $\cH^s_n$, $\vec \cH^s$ and $\cH^s_0$ be the Sobolev spaces of order $s$ associated with the $L^2$-spaces 
 \begin{align}\label{L2n-space} &L^{2}_{n}:=\{\psi\in L^2_{\rm loc}(\R^2, \C): 
 \psi\ \text{ satisfies \eqref{psi-gaugeper}}\},\\  
\label{vecL2-space} & \vec L^{2}_{\divv,0}:=\{\ \al\in L^2_{\rm loc}(\R^2, \R^2)\ |\ \ \al \text{ satisfies  \eqref{alpha-per}} 
 \}, \\  
\label{L20-space} & L^{2}_{0}:=\{a_0\in L^2_{\rm loc}(\R^2, \R)\ |\ a_0 \text{ satisfies  \eqref{a0-per} and } \lan a_0\ran =0\}, \end{align} 
 where $\divv \al$ is understood in the distributional sense, with the inner products of $L^2(\hat\Om, \C)$, $L^2(\hat\Om, \R^2)$ and $L^{2}(\hat\Om, \R)$, for some fundamental cell $\hat\Om$ of $\hat\lat$. 
 Let  \begin{align}\label{XYspaces}
 X:=\cH^{2}_{n} \times \vec\cH^{1} \times \cH^{1}_{0}\times \R\ \text{ and }\ Y:= L^{2}_{n}\times \vec L^{2}_{\divv,0} \times L^{2}_{0}\times \R.\end{align} 

Let $P'$ be the orthogonal projection onto the divergence free 
 vector fields ($P'=\frac{1}{-\Delta}\curl^*\curl$) and $u:=(\psi,  \al, a_0, \theta)$, where $\al:=  a - a^n$ and  $\theta:=\lam-\frac nb$,  so that the normal state is given by $u=0$. 
   We introduce the map $F_{n b}(u): X \ra Y$, related to the l.h.s. of \eqref{CSeqs-stat-resc}  as 
 \begin{align}\label{F}
  F_{n b} (u):=  \left( \begin{array}{cccc} - \Delta_{ a} \psi + 
v_\lam (|\psi|^2)   \psi- a_0 \psi\\  
 \curl^*a_0- P' J (\psi,  \al) \\  
 \curl \al - \frac12 \avg{\Abs{\psi}^2} + \frac12 \Abs{\psi}^2\\ 
 - \theta b +\frac12  \lan |\psi|^2\ran\end{array} \right)
\end{align}
where $J (\psi, \al) :=\im (\bar \psi \n_{ a_n+ \al} \psi)$ and, recall,  $v_\lam (|\psi|^2):=\lam V'(\lam^{-1}|\psi|^2)\equiv \lam v(\lam^{-1}|\psi|^2)$. We see that the map $F_{n b}(u)$ differs from the l.h.s. of \eqref{CSeqs-stat-resc} by having the projection $P'$ in the second entry, all other entries are unchanged. This is to make sure that $\Ran F_{n b}(u)\subset Y$.

Exactly as in a result in \cite{TS1}, we have

\begin{proposition}\label{prop:reconstr}  Assume $(\psi, \al, a_0, \theta)$ is a solution of the equation 
\begin{equation}\label{Feq} F_{n b}(\psi, \al, a_0, \theta) = 0, \end{equation}
 with $\lam = \frac nb +\theta$, 
 satisfying \eqref{gaugeper'}. 
Then $\divv J (\psi, \al)=0$ 
and therefore $(\psi, \al, a_0, \theta)$ solves  system \eqref{CSeqs-stat-resc}.  
   \end{proposition}
 \begin{proof} 
    Assume 
 $\chi\in H^1_{\rm loc}$ and is $\lat-$periodic (we say, $\chi \in H^1_{\rm per} $). Following \cite{TS1}, we differentiate the equation $\cE_\lam(e^{i s\chi}\psi, \al+s\nabla\chi, a_0)=\cE_\lam(\psi, \al, a_0) $, w.r.to $s$ at $s=0$, use that $\curl\n \chi=0$  and integrate by parts, to obtain
\begin{align} \label{E-deriv}
 \re\lan (-\Delta_{a^{n}+ \alpha} -a_0)\psi  + v_\lam (|\psi|^2)\psi, & i \chi\psi\ran 
 + \lan 
J (\psi, \al), \n \chi\ran=0.\end{align} 
 (Due to conditions \eqref{gaugeper'} and the $\hat\lat-$periodicity of $\chi$, 
 there are no boundary terms.) 
This, together with the first equation in \eqref{Feq}, 
implies \begin{align}\label{J-relat} \lan J (\psi, \al), \n\chi\ran=0.\end{align} 
 Since the last equation holds for any  $\chi \in H^1_{\rm per} $, we conclude that $\divv J (\phi,  \al)=0$.
\end{proof}

 We conclude this section by establishing some general properties of the map $F_{n b} : \R \times X \to Y$ used below. 
We use the obvious notation $F_{n b} =(F_{n b 1}, F_{n b 2}, F_{n b 3}, F_{n b 4})$. For $f=(f_1, f_2, f_3, f_4)$, we introduce the gauge transformation as $T_\del f= (e^{i\delta} f_1, f_2, f_3, f_4)$. The following proposition lists some properties of $F$.
 
 \begin{proposition}\label{prop:F-propert} Recall the notation $u:=(\psi,  \al, a_0, \theta)$. We have
    \hfill
    \begin{enumerate}[(a)]
    \item for all $b$, $F_{n b}(0) = 0$,
    \item \label{F-gauge-inv} for all $\delta \in \R$, $F_{n b} ( T_{\delta}u) = T_\del  F_{n b} ( u)$,
	\item \label{psiF-real} for all $u$ (resp. $\psi$), $\langle u, F_{n b} (u) \rangle \in \R$ (resp. $\langle \psi, F_{n b 1} (u) \rangle \in \R$).
    \end{enumerate}
\end{proposition}
\begin{proof}
In this proof, we omit the subindex $n b $.  The properties   (a) and (b) are straightforward calculations using the definition of $F$. For (c), since $\langle u, F \rangle =\langle \psi, F_1 \rangle +\langle \alpha, F_2 \rangle +\langle a_0, F_3 \rangle + \theta F_4 $ and $\langle \alpha, F_2( u) \rangle, \langle a_0, F_3 \rangle $ and $ \theta F_4$ is real, the statements $\langle u, F( u) \rangle \in \R$ and $\langle \phi, F_1( u) \rangle \in \R$ are equivalent. Now, expanding $\Delta_{\an+\al}$ in $\al$, we calculate 
   	\begin{align*}
		\langle \psi, F_1( u) \rangle
		&= \langle \psi, (-\Delta_{\an} - \frac nb \chi)\psi \rangle + 2i\int_{\Omn} \bar{\psi}\al \cdot\nabla\psi\\
			&+ 2\int_{\Omn} (\alpha \cdot a^n)|\psi|^2
			+ \int_{\Omn} |\al |^2 |\psi|^2\\
			&			+ \int_{\Omn} (v_\lam (|\psi|^2)  + \frac nb \chi- 
			 a_0 ) |\psi|^2.			
	\end{align*}
				The final three terms are clearly real and so is the first because $L^n - \lambda$ is self-adjoint. For the second term, we integrate by parts and  use the fact that the boundary terms vanish due to the periodicity of the integrand and that $\divv \al = 0$ to see that
	\begin{equation*}
\im 2i\int_{\Omn}  \psi \alpha \cdot\nabla\bar{\psi}			= \int_{\Omn}\al \cdot (\bar{\psi} \nabla\psi + \phi  \nabla\bar{\psi} )
			= -\int_{\Omn} (\divv \al) |\psi|^2=0.		
	\end{equation*}
	 Thus this term is also real and (c) is established.
\end{proof}

\section{Linearized system} \label{sec:linear}

In this section we use the Lyapunov-Schmidt decomposition to reduce the problem of solving  equation  \eqref{Feq}  to a finite dimensional problem. We address the latter in the next section.

Recall that $\chi:=-v(0)$ and let  $u:=(\psi, \al, a_0, \theta)$. In this notation,  the normal (`trivial') solution is given by $u^{n b}=0$. The linearized map on it,  
 $ A_{n b}:= d F_{n b}(0)$,  is given explicitly by

\begin{equation}\label{Anb} A_{n b} = \left( \begin{array}{cccc} - \Delta_{ a^n} - \frac nb \chi & 0 & 0 &0\\  0 & 0 &  \curl^* & 0 \\  0 &  \curl &0 & 0 \\  0 & 0 & 0 &  -b \end{array} \right).\end{equation}			%
Note that $A_{n b}$ maps $X$ into $Y$.  We have
\begin{proposition}\label{prop:dFspec}  The operator $A_{n b}$ is self-adjoint and has purely discrete spectrum accumulating at $\pm\infty$.  Moreover, in spaces \eqref{L2n-space} - \eqref{L20-space}, 
\begin{align}\label{null-dF} &\Null A_{n b} = \Null (- \Delta_{a^n} - n\frac \chi b) \times \{const\} \times \{0\} \times \{0\},\\ 
 \label{null-Delta-ab}&\Null (- \Delta_{a^n} - n\frac \chi b )\neq \{0\} \Longleftrightarrow b = \chi \quad (b \text{ close to } \chi),\\ 
\label{null-Delta-ab'} & \dim \Null (-  \Delta_{a^n} - n)=n.\end{align} 
   \end{proposition}
 \begin{proof} 
The self-adjointness of $A_{n b}$ follows from standard results. Eqs \eqref{null-Delta-ab} - \eqref{null-Delta-ab'} are known, see e.g. \cite{TS1, GS} (see also Corollary \ref{cor:LaplAb-spec} of  Appendix \ref{sec:self-dual}). 
It remains to prove \eqref{null-dF}. Let
\begin{equation} \label{Mop}
	M = \begin{pmatrix}
	0 &  \curl^*  \\
	  \curl  & 0
	\end{pmatrix}
\end{equation}
Eq \eqref{Anb} shows that $A_{n b} =  (- \Delta_{a^n} - n\frac \chi b) \oplus\Null M \oplus - b$ and therefore 
\begin{align}\label{null-dF'} &\Null A_{n b} = \Null (- \Delta_{a^n} - n\frac \chi b) \times\Null M \times \{0\}.\end{align}
For the operator $M$, we have the following

\begin{lemma} \label{lem:Mspec}
	The operator $M$ on $L^2_{per}(\R^2, \R^2) \times L^2_{per}(\R^2, \R)$, with domain $H^1_{per}(\R^2, \R^2) \times H^1_{per}(\R^2, \R)$ is self-adjoint and has 
	the spectrum \[\{0, \pm |k|: k\in \latn^*\},\]
	 where $\lat^*$ is the lattice reciprocal to $\lat$, with the corresponding eigenfunctions $( k_1, k_2, 0)e^{i k\cdot x}$ (or  $( \n f, 0)$) and $(\mp  \hat k_2, \pm  \hat k_1, 1)e^{i k\cdot x}$, with $\hat k_:=k_j/|k|$, respectively. In particular,	
	\begin{equation} \label{NullM}
			\Null M = \R \times \R^2
	\end{equation}
\end{lemma}
\begin{proof}
 The self-adjointness of $M$ is standard. One can easily check that  $(k_1, k_2, 0)e^{i k\cdot x}$ (or  $(\n f, 0)$) and $(\mp  \hat k_2, \pm  \hat k_1, 1)e^{i k\cdot x}$ are eigenfunctions and that they form a basis in $L^2_{per}(\R^2, \R^2) \times L^2_{per}(\R^2, \R)$. 
 \end{proof}
Equations \eqref{null-dF'} and \eqref{NullM} imply  \eqref{null-dF}. Eqs. \eqref{null-Delta-ab} and \eqref{null-Delta-ab'} are standard (see e.g. \cite{TS1}).
 \end{proof}
\paragraph{Remarks.} 1) The operator \eqref{Anb}  is the hessian of the energy functional \eqref{en'} at the point $(b, 0)$. 

2) Usually, the linearization of the equations at a solution determines its linear stability. And with this criterion,  the normal state $u^{nb}$ seems to to be always unstable. However,  equations \eqref{CSeqs-stat-resc2}-\eqref{CSeqs-stat-resc4} can be considered as constraints and therefore only the $\psi$-block determines the linear stability of the normal state $u^{nb}$. Hence equation \eqref{Anb} and the standard result $- \Delta_{a^n}\ge n$ (see e.g. \cite{TS1}) show that $u^{nb}$ is linearly stable if $b\ge  n\chi$ and unstable for $b< n \chi$. (See \cite{S}, for more details.)


\section{Reduction to a finite-dimensional problem} \label{sec:reduction}

   We let $P$ be the orthogonal projection in $Y$ onto $K := \Null_{X } A_{n \chi}\subset X$ (see \eqref{null-dF}) and let $\bar P := I - P$.  
\DETAILS{Since $ n$ is an isolated eigenvalue of $A_{ n}$, $P$ can be explicitly given as the  Riesz projection,
    \begin{equation}
        P := -\frac{1}{2\pi i} \oint_\gamma (A_{ n} - z)^{-1} \,dz,
    \end{equation}
    where $\gamma \subseteq \C$ is a contour around $0$ that contains no other points of the spectrum of $A_{ n}$.} 

    Writing $u=v +w$, where $v = P u$ and $w = \bar P u$, we see that the equation $F_{n b}'(u) = 0$ is therefore equivalent to the pair of equations
    \begin{align}
        \label{LS:eqn1} &P F_{n b} (v + w) = 0, \\
        \label{LS:eqn2} &\bar P F_{n b} (v + w) = 0.
    \end{align}
   We will now solve \eqref{LS:eqn2} for $w = \bar P u$ in terms of $b$ and $v = P u$. Let $s:=\frac{1}{\|\psio\|^2} \langle \psio, \psi \rangle$   and, recall, $ \lan  \al\ran:=\frac{1}{|\Om|}\int_{\Om} \al$.   We have
   \begin{lemma}\label{w-est} There is a neighbourhood, $U\subset \R \times K$, of $(n, 0)$,
    such that for any $(b, v)$ in that neighbourhood, Eq \eqref{LS:eqn2} has a unique solution $w = w(b, v)$. This solution satisfies 
  \begin{align}\label{w-prop1}
&  w(b, v)  \ \mbox{real-analytic in}\    (b, v),\\ 
   \label{w-prop2}
   &   \|\p_b^m w_i\|_{X_i}=O(\|v\|_{X_i}^{2}),\ i=  1, 2, 3, 4,\ m=0, 1,\\ 
  \label{w1-est}
      &  \|\p_b^m (w_1-w_1')\|_{X_1}=O(\|v\|_{X_i}^{3}),\ 
      m=0, 1, \end{align}
      where $w(b, v)=(w_1, w_2, w_3, w_4)$,  $w_1':=-2i s \lan  \al\ran \cdot (- \Delta_{ a^n} - \frac nb \chi)^{-1}\nabla_{\an}\psio$ and $Y_i$ are the factors in space $X$ defined in  \eqref{XYspaces}. 

     \end{lemma}
    \begin{proof} 
    We introduce the map
    $G : \R \times K \times \bar X \to \bar Y$, where  $\bar X := \bar P X=X\ominus K$ and $\bar Y :=\bar P Y= Y \ominus K$, defined by  \begin{align}\label{G}G(b, v, w) =\bar P F_{n b} (v + w).\end{align} 
 It has the following properties (a) $G$ is $C^2$; (b) $G(b, 0, 0)=0$ $\forall b$; (c) $d_w G(b, 0, 0)$  is invertible for $b =n$.      Applying the Implicit Function Theorem
    to $G=0$, we obtain a function $w : \R \times K \to \bar X$, defined on a neighbourhood of $(\chi, 0)$,
    such that $w = w(b, v)$ is a unique solution to $G(b, v, w) = 0$, for $(b, v)$ in that neighbourhood. This proves the first statement.
    
    By the implicit function theorem and the analyticity of $F $, the  solution has the  property \eqref{w-prop1}.
    
To prove estimates \eqref{w-prop2} - \eqref{w1-est}, we recall $u:=(\psi, \al, a_0, \theta)$ and $A_{n b}:= d F_{n b}(0)$ 
 and write $F_{n b}$ as 
\begin{align}\label{F-deco}
 & \qquad  F_{n b} (u) = A_{n b} u +f_{n b} (u) ,\\
\DETAILS{ \begin{equation}\label{f}
 f(u):=     \begin{cases}
(V'_\lam (|\psi|^2)  -\chi) \psi- \lam a_0 \psi,\\  
\lam^{-1}P'\im (\bar \psi \n_{\vec a_n+ \al} \psi),\\  
- |\psi|^2 + \lan |\psi|^2\ran =0,\\ 
 n - \lam^2 b -\lan |\psi|^2\ran. 
  \end{cases} \end{equation}}
 \label{f}
&  f_{n b} (u):=  \left( \begin{array}{cccc} h_{n b} (u)\\  
-  \lam^{-1}P'J (\psi, \al)\\  
- \frac12 |\psi|^2 + \frac12 \lan |\psi|^2\ran\\ 
 - \theta^2 b-\frac12 \lan |\psi|^2\ran\end{array} \right),
\end{align}
with, recall,  $\lambda= \frac nb+\theta$, $J (\psi, \al):=\im (\bar \psi \n_{\vec a_n+ \al} \psi)$  and, with  $v_\lam (|\psi|^2):=\lam V'(\lam^{-1}|\psi|^2)\equiv \lam v(\lam^{-1}|\psi|^2)$,  
\begin{align}\label{h-def}h_{n b} (u) :=(v_\lam (|\psi|^2)  + \frac nb \chi ) \psi-  a_0 \psi+ 2i\al \cdot\nabla_{\an}\psi 
 +  |\al |^2 \psi.\end{align}
 
Proposition \ref{prop:dFspec} above and a standard spectral theory imply that $A_{n b}^\perp:= \bar P A_{n b}  \bar P\big |_{\Ran  \bar P}$ is invertible for $b$ close to $n $, with the uniformly bounded inverse. Hence,   using decomposition \eqref{F-deco}, we can rewrite  \eqref{LS:eqn2} as $A_{n b}^\perp w = -  \bar P f(b, u)$, which  after inverting $A_{n b}^\perp$ becomes
\begin{align}\label{w-eq}
   w = - (A_{n b}^\perp)^{-1} \bar P f_{n b}(u).
\end{align}
  We use that $A_{n b}$ is diagonal, $f$ is of the form \eqref{f}, standard Sobolev inequalities with the standard notation for the Sobolev spaces and the definition of the space $X$ in \eqref{XYspaces} to estimate $\|w_1\|_{H^2}\ls \| u \|_{X}^2$, $\|w_2\|_{H^2}\ls \| J(u) \|_{H^1} \ls \| u \|_{H^2}^2$,  $\|w_3\|_{H^2}\ls \| u \|_{X}^2$  and $|w_4|\ls \| u \|_{X}^2$. These estimates and the triangle inequality $ \| u \|_{X}^2\ls  \| v \|_{X}^2+ \| w \|_{X}^2$ result in nonlinear inequalities giving \eqref{w-prop2}.\footnote{This result could be also derived from the proof of the implicit function theorem. Indeed,  at the core of the proof is the fixed point problem set at the product of the balls $ \| w_i\|_{X_i}\le R_i$ in  $\bar X := \bar P X$. Choosing the radii of the ball to be proportional to $\|v\|_{X_i}^{2}$ would give the desired estimate.}

For \eqref{w1-est}, using  that $A_{n b}$  and $f$ are of the form \eqref{Anb}  and \eqref{f}, we find $w_1=-(- \Delta_{ a^n} - \frac nb \chi)^{-1} h(u)$. 
 Furthermore, since $v_\lam (0)  = -\lam \chi$ and $\lambda= \frac nb+\theta$, we have that $v_\lam (|\psi|^2)  + \frac nb \chi=- \chi(\lam- \frac nb) + O(|\psi|^2)=- \chi \theta+ O(|\psi|^2)$. Next, by \eqref{null-dF}, the partition of unity, $P+\bar P=\one$, breaks $\psi$ and $\al$ as $\psi=s \psio+ \psi^\perp$  and $\al=\lan  \al\ran + \al^\perp$, where, recall,  $s:=\frac{1}{\|\psio\|^2} \langle \psio, \psi \rangle$ and $ \lan  \al\ran:=\frac{1}{|\Om|}\int_{\Om} \al$, and $\psi^\perp$ and $\al^\perp$ are defined by this relations (see \eqref{Pu-expl} for more details). 
We use these decompositions, \eqref{h-def} and standard Sobolev inequalities to estimate 
\begin{align}\notag \| h(u) -  h_1(u)\|_{L^2}\ls & \| u \|_{X}^3+ |\theta| \| \psi \|_{L^2} + |\lan  \al\ran| \| \psi^\perp \|_{H^1}\\
\notag & + \| \al^\perp \|_{H^1}\| \psi \|_{H^2}+ \| a_0 \|_{H^1}\| \psi \|_{H^2}\end{align}
where $h_1(u):=2i s \lan  \al\ran \cdot\nabla_{\an}\psio$. Since $\psi^\perp, \al^\perp, a_0$ and $\theta$ contribute only to $w:=\bar P u$, this, together with $w_1=-(-  \Delta_{ a^n} - \frac nb \chi)^{-1} h(u)$, gives  
\begin{align}\label{w1-h-est'}\|w_1-w_1'\|_{H^2}\ls \| (h(u) -  h_1(u))\|_{L^2}\ls   
 \| u \|_{X}^3+\| w \|_{X}\| u \|_{X}\end{align}
where, recall, $w_1':=-2i s \lan  \al\ran \cdot (- \Delta_{ a^n} - \frac nb \chi)^{-1}\nabla_{\an}\psio$. Finally, taking into account $\|w_2\|_{H^2}\ls \| u \|_{X}^2$,  $\|w_3\|_{H^2}\ls \| u \|_{X}^2$  and $|w_4|\ls \| u \|_{X}^2$, we obtain 
\begin{align}\label{w1-h-est}\|w_1-w_1'\|_{H^2}\ls \| (h(u) -  h_1(u)) \|_{L^2}\ls |s|  |\lan  \al\ran| \| \psi \|_{H^1}+ \| u \|_{X}^3.\end{align}
  \DETAILS{and expression \eqref{h-def} and the definition of the space $X$ in \eqref{XYspaces} to estimate standard Sobolev inequalities with the standard notation for the Sobolev spaces to conclude that   $\|w_1\|_{H^2}\ls \| h(u) \|_{L^2} \ls \| u \|_{H^2}^3$, $\|w_2\|_{H^2}\ls \| J(u) \|_{H^1} \ls \| u \|_{H^2}^2$,   and $\|w_3\|_{H^2}\ls \| u \|_{H^2}^2$.}
   Since $|s|, |\lan  \al\ran|, \| \psi \|_{H^1}\ls \| v \|_{X}$ and $u=v+w$, with $\| w \|_{X}\ll \| v \|_{X}$, this gives 
 \eqref{w1-est} with $m=0$. The  \eqref{w1-est} with $m>0$ is proven similarly.  
 \end{proof}

    We substitute the solution $w = w(b, v)$ into \eqref{LS:eqn1} and see that the latter equation 
     in a neighbourhood of $(\chi, 0)$ is equivalent to the equation  (the \emph{bifurcation equation})
           \begin{equation}  \label{bif-eqn}
        \gamma(b, v):= PF_{n b} (v + w(b, v)) = 0.
    \end{equation}
    Note that 
     $\gamma : \R \times K \to K$. 
  Thus we have shown the following
  \begin{corollary}\label{cor:reduct-fd}
Let $ u=(\psi, \al, a_0, \theta)$.  In a neighbourhood of $(\chi, 0)$ in $\R \times X$, the pair $(b, u)$ solves  Eq  \eqref{Feq}  (and therefore \eqref{CSeqs-stat-resc}),     if and only if $(b, v)$, with $v = P u$, solves
    \eqref{bif-eqn}. Moreover, the solution $u$ of \eqref{Feq} can be reconstructed from the solution $v$ of \eqref{bif-eqn} according to the formula
     \begin{equation} \label{uvw}
     u =v+ w(b, v),
    \end{equation} 
where $w = w(b, v)$ is the unique solution to Eq \eqref{LS:eqn2} described in Lemma \ref{w-est}.   
    \end{corollary}

Solving the bifurcation equation \eqref{bif-eqn} is a subtle problem. 
We do this in Section \ref{sec:bifurc-dimK=1}.


\section{Existence result}  \label{sec:bifurc-dimK=1}

\begin{theorem} \label{thm:bif-thm-Kdim1}
Assume  $n=1$. 
Then, 
for every $\latn$,     there exist $\epsilon > 0$ and a branch, $ (b_s, u_s),\ u_s:=  (\psi_s, \al_{s}, a_{0 s}, \theta_s)$, $s \in [0, \sqrt{\e})$, 
    of nontrivial, 
    $\hat\lat$-equivariant solutions of  the rescaled ZHK system \eqref{CSeqs-stat-resc} bifurcating from $(\chi, 0)$.  
    It is unique modulo the global gauge symmetry 
     in a sufficiently small neighbourhood of $(\chi, 0)$ 
     in $\R \times X$, 
    and satisfies
    \begin{equation}\label{s-expans1}
    \begin{cases}
        b_s =  \chi + O (s^2),\\
       \psi_s = s\psio + O_{\cH^{2}_{n}}(s^3),\\ 
  \al_{s}=    O_{\vec\cH^{2}}(s^2),\\ 
     a_{0 s} =   O_{\cH^{2}_{0}}(s^2),\\ 
      \theta_s =    
       O (s^2), 
            \end{cases}
    \end{equation}
    where $\psio$, 
   normalized as $\lan|\psio|^2\ran =1$, satisfies the equation
 \begin{align} \label{psi0-eq} (- \Delta_{ a^n} -  n) \psio=0  \end{align}
    
    \end{theorem}
 \begin{proof} 
 As 
 in the last section, $X =\cH^{2}_{1}\times \HA{2}{\Div,0}\times \cH^{2}_{0}\times \R$ and $Y = \Lpsi{2}{1}\times \LA{p}{\Div,0}\times L^{2}_{0}\times \R$.
 Our goal is to solve the equation \eqref{bif-eqn} for $v$.  
First, we  prove the gauge invariance of $\gamma$: 

    \begin{lemma}\label{lem:w-gam-gaugeinv}
        For every $\delta \in \R$, $w(b, e^{i\delta}v) = T_\del w(b, v)$ and $\gamma(b, e^{i\delta}v) = e^{i\delta} \gamma(b, v)$.
    \end{lemma}
    \begin{proof}
        We first check that $w(b, e^{i\delta}v) = T_\del w(b, v)$. We note that by definition of $w$,
    \begin{align*}G(b,  e^{i\delta} v, w(b, e^{i\delta} v)) = 0, \end{align*} 
 where $G$ is defined in \eqref{G},   but by property  in Proposition \ref{prop:F-propert}\eqref{F-gauge-inv}, 
     we also have
        \[G(b,  e^{i\delta} v,  e^{i\delta} w(b, v)) = T_\del  G(b,v, w(\lambda,v)) = 0.\] The uniqueness of $w$
        then implies that $w(b, e^{i\delta} v) = T_\del w(b, v)$.
        Using that $e^{i\delta} v=T_\del v$, we can now verify that
        \begin{align*}
            \gamma(b, e^{i\delta} v) = PF_{n b} (e^{i\delta} v + w(b,  e^{i\delta} v))&= PF_{n b} (T_\del (v + w(\lambda,  v)))\\
                &= PT_\del F_{n b} (v + w(b,  v)).         \end{align*}
 Since  
$ P$ is of the form $P=P_1\oplus 0$, where $P_1$ acts  on the first component, we have $PT_\del F_{n b} (v + w(b,  v))= e^{i\delta} PF_{n b} (v + w(b, v) ) = e^{i\delta}\gamma(b,v)$, which implies the second statement.   \end{proof}

    By Proposition \ref{prop:dFspec},  we have 
   \begin{align}            \label{NullA-NullL}  \Null_{X } A_{n \chi}= \Null_{\cH^{2}_{1} } (-\Delta_{\an}-n)\times \{\text{const}\} \times \{0\} \times \{0\}. 
\end{align}    
According to \eqref{null-Delta-ab} - \eqref{null-Delta-ab'}, the assumption $n=1$ gives
\begin{align} \label{Null-dim1-cond}  \dim_{\C}\Null_{\cH^{2}_{n}} (-\Delta_{\an}-n)=1. 
\end{align}   
 This relation and   \eqref{NullA-NullL} 
  yield that the projection $P$  can be written, for $u=(\psi, \al, a_0, \theta)$, as 
    \begin{align}    \label{Pu-expl}
        P u 
        = (s \psio, \lan  \al\ran, 0, 0)\ \textrm{with}\ & s:=\frac{1}{\|\psio\|^2} \langle \psio, \psi \rangle,\\ 
        \notag &\psio \in \Null_{\cH^{2}_{1}} (-\Delta_{\an} - n),\ \|\psio\|=1, 
    \end{align}
where, recall, $ \lan  \al\ran:=\frac{1}{|\Om|}\int_{\Om} \al$. Hence,  we can  write the map $\gamma$ in the bifurcation equation \eqref{bif-eqn} as $\gamma=(\psio \g_1, \g_2, 0, 0)$, where $\gamma_i : \R \times \C \to \C$ are given by   
\begin{align} \label{tildegam1}
&\gamma_1(b, s, \mu) := \langle \psio, F_{n b 1} (v + w(b, v)) \rangle\\
\label{tildegam2} &\gamma_2(b, s, \mu) := \langle F_{n b 2} (v + w(b, v)) \rangle, \end{align}
with $(F_1,  F_2, F_3,  F_4)=F$,  $\mu :=\lan  \al\ran/s\in \R^2$ and $v = v_{s, \mu}$, where  
 \begin{align}    \label{vsmu}v_{s, \mu}:=(s\psio, s\mu, 0, 0). \end{align}  Now,  equation \eqref{bif-eqn} is equivalent to the equations 
\begin{align} \label{gami-eqs}
&\gamma_1(b, s, \mu) =0,\ \qquad \gamma_2(b, s, \mu) =0. \end{align}

First, we consider the equation $\gamma_2(b, s, \mu)=0$. By the second equation in \eqref{F} - \eqref{Feq} 
 and the relations $\lan P' J(\psi, \al) \ran= \lan \one, P' J(\psi, \al) \ran=\lan P' \one, J(\psi, \al) \ran$ and $P'  \one =\one$, we have
   \begin{equation} \label{gam2} \gamma_2(b, s, \mu)= \lan P' J(v + w(b, v)) \ran= \lan J(v + w(b, v)) \ran.
\end{equation}
According to \eqref{w1-est} and $\lan  \al\ran=s \mu$, we write $w=w'+w''$, where $w':=(- i s^2\mu  \cdot \om_1, 0, 0, 0)$, with  $\om_1:=(-  \Delta_{ a^n} - \frac nb \chi)^{-1}\nabla_{\an}\psio$. Then, by \eqref{w-eq}, \eqref{uvw}, \eqref{vsmu}, \eqref{w-prop2} and \eqref{w1-est}, we have
 \begin{align}\label{w-est''}\p_\xi^m w''(b, v)=(O_{\cH_{n}^2}(s^3), O_{\HA{2}{\Div,0}}(s^2), O_{\cH_{0}^2}(s^2), O(s^2)), \end{align} 
 for $m=0, 1$ and $\xi=b, \mu$. 
 Hence, we have for  $J (\psi, \al) := \im (\bar \psi \n_{ a_n+ \al} \psi)$,
 \begin{align}\label{J-expan} 
 J(v + w(b, v))=&  s^2\im (\bar \psio \n_{a^n} \psio) +  s^3 (|\psio|^2 \mu +T\mu)\notag\\ &
+O_{\cH^{2}_{n}}(s^4)+O_{\HA{2}{\Div,0}}(s^4),  \end{align} 
where  $T \mu = \im (\bar \psio \mu \cdot \n_{ \an}  i\om_1 
 -i \bar \om_1 \mu \cdot \n_{\an}  \psi_0 )$.

 Recall, that $\curl^*$ maps scalar functions, $f,$ into vector-fields, $\curl^* f = (\p_2 f, -\p_1 f)$. The next lemma, due to \cite{TS1}, shows that the leading term on the r.h.s. of \eqref{J-expan} 
drops out under taking the average.
\begin{lemma}\label{lem:current-repr}  
\begin{equation} \label{eq:J.co.exact} 
  \im (\bar\psio \nabla_{a^n}  \psio) = \frac{1}{ 2} \curl^* |\psio |^2.
\end{equation}
  \end{lemma}
 For the reader's convenience, the proof of this lemma is given in Appendix \ref{sec:current-repr}. (Due to the different sign in the covariant gradient  $\n_{ A}$, the sign here differs from the one in \cite{TS1}. Note that, unlike the Ginzburg-Landau equation,s the ZHK equations are not  invariant under the transformation $(\psi, A) \ra (\bar\psi, -A)$.) 
 By \eqref{eq:J.co.exact}, $\lan \im (\bar\psio \nabla_{a^n}  \psio)\ran =0$. 
This, the definition  $\tilde\gamma_2(b, s, \mu) := s^{-3} \gamma_2(b, s, \mu)$, \eqref{gam2}, \eqref{J-expan} and $\lan |\psio|^2\ran = 1$ give 
\begin{equation} \label{tilde-gam2-exp}\tilde\gamma_2(b, s, \mu)=   (\one + \lan T\ran)\mu  +O(s).\end{equation} 
 Furthermore, $\tilde\gamma_2(b, 0, 0) =0$ (for any $\lam$) and it is easy to see that $\p_\mu\tilde \gamma_2(b, s, \mu)=\one + \lan T\ran+O(s)$. Next, we claim that 
\begin{align}\label{Tinv} 	 \text{The matrix $\lan T\ran$ is positive definite and therefore $\one + \lan T\ran$ is invertible.}
 \end{align} 
 To show that the matrix $\lan T \ran$ is positive definite, we note first that  
 $T_{ij}:=-\re (\bar \psio \n_{ \an j}  \om_{1 i} - \bar \om_{1 i} \n_{ \an j} \psio)$, with $\om_1:=(- \Delta_{ a^n} - \frac nb \chi)^{-1}\nabla_{\an}\psio$.
 Using this, we compute 
  \begin{align}\lan T_{ij} \ran &:=-\re \lan \bar \psio \n_{ \an j}  \om_{1 i} - \bar \om_{1 i} \n_{ \an j} \psio \ran\\
 	&=-\re [\lan  \psio, \n_{ \an j} \om_{1 i} \ran - \lan   \om_{1 i}, \n_{ \an j} \psio \ran]\\
 	&=2\lan  \n_{ \an j} \psio, \om_{1 i} \ran. \end{align}
Since  $\om_1:=(- \Delta_{ a^n} - \frac nb \chi)^{-1}\nabla_{\an}\psio$, this gives 
\begin{align}\lan T_{ij} \ran=2\lan  \n_{ a_n j} \psio, (- \Delta_{ a^n} - \frac nb \chi)^{-1}\nabla_{\an i}\psio \ran. \end{align} 
   Now since $\langle \psi_0, \nabla_{a^n} \psi_0 \rangle = 0$, and $(- \Delta_{ a^n} - \frac nb \chi)^{-1}$ is positive for vectors orthogonal to $\psi_0$, the last relation shows that the matrix $\lan T \ran$ is positive definite.

   By \eqref{tilde-gam2-exp} and \eqref{Tinv}, %
  the equation $\tilde \gamma_2(b, s, \mu)= 0$ has a unique solution, $\mu_s$, for $\mu$, provided $s$ is sufficiently small, and this solution satisfies 
\begin{align} \label{mus-expan}& \mu_s=O(s). \end{align} 

 Estimates \eqref{w1-est} and \eqref{mus-expan} and the relation  $\lan  \al\ran= s\mu$ give $\|\p_b^m w_1\|_{X_1}=O(s^{3}),\ m=0, 1,$ and we can upgrade \eqref{w-est''} to 
\begin{align}\label{w-est'}\p_\xi^m w(b, v)=(O_{\cH_{n}^2}(s^3), O_{\HA{2}{\Div,0}}(s^2), O_{\cH_{0}^2}(s^2), O(s^2)). \end{align}

Next, we address the equation $\gamma_1(b, s, \mu)=0$. First, we show that $\gamma_1(b, s) \in \R$ for $s \in \R$.
Using that the projection $\bar P$ is self-adjoint,  $\bar P w(b, v) = w(b, v)$ and that $w(b, v)$ solves  $ \bar P F_{n b} (v + w)=0$, we find
	\begin{align*}
		\langle w(b, v), F_{n b} (v+ w(b, v)) \rangle
		= \langle w(b, v), \bar P F_{n b} (v + w(b, v)) \rangle
		= 0.
	\end{align*}
	Therefore, recalling $v\equiv v_{s, \mu}:=(s\psio, s\mu, 0, 0)$ and using that $\langle v, F\rangle=s\langle \psio, F_1 \rangle + 
	s\mu\langle F_2 \rangle$ and $\langle F_{n b 2} (v + w (b, v)) \rangle=- \gamma_2(b, s, \lan  \al\ran)=0$, we have, for $s \neq 0$,
	\begin{align*}
		\langle \psi_0, F_{n b 1} (v + w(b, v)) \rangle
	&= s^{-1}\langle v, F_{n b} (v + w (b, v)) \rangle	\\ 
	&= s^{-1}\langle v + w (b, v), F_{n b} (v + w (b, v)) \rangle,
	\end{align*}
	and this is real by property \eqref{psiF-real} in Proposition \ref{prop:F-propert} and the fact that the part $\langle  w_2 (b, v),$ $ F_{n b 2} (v + w (b, v)) \rangle$ of the inner product on the r.h.s.  is real. 
	
	Next, by Lemma \ref{lem:w-gam-gaugeinv}, $\gamma_1(b, s, \mu) = e^{i\arg s} \gamma_1(b, |s|, \mu)$. Therefore $\gamma_1(b, s, \mu)=0$ is equivalent to the equation
\begin{equation} \label{bif-eqn1'} \gamma_1'(b, s, \mu) =0
\end{equation}
for the restriction $\gamma_1' : \R \times \R \times \R \to \R$ of the function $\gamma_1$ to $\R \times \R \times \R$, i.e., for real $s$. 

Now, recall that the map $F =(F_1, F_2, F_3, F_4)$, defined in \eqref{F}, 
can be written as \eqref{F-deco}.
Using $w(\lambda, v)=O(s^3)$,  \eqref{w1-h-est},  $\lan  \al\ran = s\mu$, by the definition $\mu :=\lan  \al\ran/s$ and \eqref{mus-expan}, and recalling $v\equiv v_{s, \mu}:=(s\psio, s\mu, 0, 0)$,  we  find  \begin{align}\label{F1-expan} 
& F_{n b 1} (v+ w(b, v)) =  s(-\Delta_{a^n}   - \frac nb \chi) \psio +O_{\cH_{n}^2}(s^3). 
 \end{align}
Using  this, $(-\Delta_{a^n} - n) \psio=0$ 
 and denoting the first component of $P$  by $P_1$, we obtain 
\begin{align} \label{PF-expan}
\DETAILS{P F(\lambda, v) = & s (\lam\chi - n) \psi_0 
+O_{norm??}(s^3),\ 
|s|^2 \lan \al\ran + O(s^4)\big).} 
P_1 F_{n b 1} (v+ w(b, v)) = &  s (n-\frac nb\chi) \psio 
+O_{\cH_{n}^2}(s^3).\end{align}
If we write $\gamma_1'(b, s, \mu)= s \tilde\gamma_1(b, s, \mu)$, then, since $\lan |\psi_{0}|^2\ran = 1$, we have $\tilde\gamma_1(b, s, \mu)=n-\frac nb\chi+O(s^2)$. Since $\tilde\gamma_1(n, 0, \mu) =0$ (for any $\mu$) and, as easy to see, $\p_b\tilde \gamma_1(b, s, \mu)=\frac {n}{b^2}\chi +O(s^2)$, the equation $\tilde \gamma_1(b, s, \mu)= 0$  
  has the unique solution, $\lam_s$, and this solution is of the form 
 \begin{align} \label{bs-expan} 
 b_s=\chi  +  O(s^2). \end{align}

Now, we know that $(b, u)$ solves $F(b, u) = 0$ if and only if $u=v +w(b, v),\ v\equiv v_{s, \mu},$ and $b, s, \mu$ solve $\gamma(b, s, \mu) = 0$ (see \eqref{bif-eqn}), or $\gamma_i(b, s, \mu) = 0,\ i=1, 2$ (see \eqref{gami-eqs}).  By above, near $(n, 0, 0)$, 
 Eq \eqref{gami-eqs} has two branches of solutions, $s=0$ and $b, \mu$ are arbitrary and $(b_s, s, \mu_s)$, where $b_s$ and $\mu_s$ are given by \eqref{bs-expan} and \eqref{mus-expan}, respectively, and $s$ is sufficiently small, but otherwise is arbitrary. For $s = 0$, we have $v_{s, \mu}=0$ and therefore $u = v_{s, \mu} + w(b_s, v_{s, \mu}) = 0$, which gives the trivial solution. In the other case, we have, by \eqref{vsmu}, \eqref{mus-expan} and \eqref{w-est'},
 \[u = u_s:=v_{s, \mu_s} + w(b_s, v_{s, \mu_s}) = (s \psio + O_{\cH_{n}^2}(s^3),  O_{\vec \cH^{2}}(s^2), O_{\cH_0^{2}}(s^2), O(s^2)).\] 
This gives the branch $(b_s, u_s)$ of solutions of \eqref{CSeqs-stat-resc} 
   of the form \eqref{s-expans1}. These solutions satisfy \eqref{psi-gaugeper'} - \eqref{a0-per'} 
    and are unique (apart from the normal (trivial) solution $(b, 0)$) 
     modulo the global gauge symmetry 
     in a sufficiently small neighbourhood of $(b, 0)$ 
      in $\R \times X$. \end{proof}

\section{Relation between $b$ and $s$}\label{sec:b-s-relat}

To derive Theorem \ref{thm:ALexist} from Theorem \ref{thm:bif-thm-Kdim1}, we have to establish the relation between the external magnetic field $b$ and the bifurcation parameter $s$.

Here and in what follows we use  the notation $\lan f\ran:=\frac{1}{2 \pi}\int_{\hat\Omega} f$. 
To formulate the next result we introduce
 the Abrikosov function 
\begin{align}\label{beta}
	\beta(\lat) = \frac{\avg{|\psio|^4}}{{\avg{|\psio|^2}}^2}
\end{align}
where $\psio$ is the unique solution to the equation \eqref{psi0-eq}, 
  satisfying \eqref{psi-gaugeper'} for the normalization $\hat\lat$ of $\lat$ and $n=1$. By scaling invariance and \eqref{dpsio=0}, this is the same function as the one defined in \eqref{beta'}. 
Moreover, since by the H\"older inequality, $\avg{|\psio|^2}^2\le \avg{|\psio|^4}$, we have  $\beta(\lat)\ge 1$. In what follows, we   normalize $\psio$ as 
\begin{align}\label{normal-psio}\lan|\psio|^2\ran =1.\end{align}

 We  identify $\R^2$ with $\C$, via the map $(x_1, x_2)\ra x_1+i x_2$,  and, applying a rotation, if necessary, bring any lattice $\lat$ to the form  
$\lat=\lat_{\tau, r}:=r  (\Z+\tau\Z), $ 
where 
$r>0,\ \tau\in \C$, $\im\tau > 0$, which we assume from now on. In this case $\hat\lat= \Z+\tau\Z$.

Hereafter we will assume the lattice $\cL$ is of the form $\cL = r(\Z + \tau \Z)$ where $r$ and $\im \tau$ are strictly positive.  Since $\beta(\lat)$ does not depend on scaling of the lattice, $\lat \ra r\lat, r>0,$ 
 we write \begin{align}\label{beta-tau}\beta(\tau) := \beta(\lat)\ \text{ for }\ \cL = r(\Z + \tau \Z).
\end{align}

	\begin{proposition} \label{prop:lam-b-expan-alt}  	 Let $(b_s, u_s)$ be the branch of solutions of system \eqref{CSeqs-stat-resc}  given in Theorem \ref{thm:bif-thm-Kdim1}. 
Then we have the expansions  \begin{align}\label{lam-exp} 	& \lam_s = \frac n\chi+\lam' s^2 +O(s^4),
\\ \label{b-exp}	& b_s=\chi + b' s^2 +O(s^4),
	\end{align}
where $\lam'$ and $b'$ are given by 
\begin{align}\label{lam'-alt} 	&  \lam' =\frac{1}{\chi} ((g-1)\beta(\tau) +\frac{1}{2 }), 
\\ \label{b'-alt}	& b'=-\frac \chi n (g-1)\beta(\tau). 
	\end{align}
	\end{proposition}

Before proceeding to the proof of Proposition \ref{prop:lam-b-expan-alt}, we derive 
  a further expansion of the bifurcation branch given in Theorem \ref{thm:bif-thm-Kdim1} in the bifurcation parameter $s$.	

	\begin{lemma} \label{prop:ZHK.asymp}
Let $(b_s, 
 u_s)$ be the branch of solutions of system \eqref{CSeqs-stat-resc} 
  given in Theorem \ref{thm:bif-thm-Kdim1}. 
Then the last three entries in $u_s=(\psi_s, \al_{s}, a_{0 s}, \theta_s)$ are of the form 
    \begin{equation}\label{s-expans2}
    \begin{cases}
  \al_{s}=   \al' s^2 +  O_{\vec\cH^{2}}(s^4),\\
     a_{0 s} = a_{0}'  s^2 +   O_{\cH^{2}_{0}}(s^4),\\   
         \theta_s = \theta'  s^2 +         O (s^4),             \end{cases}
    \end{equation}
    where $\al', a_{0}'$ and $\theta'$ satisfy s(remembering normalization \eqref{normal-psio}, we can drop $\lan|\psio|^2\ran$ below)  
 \begin{align} \label{s-expans-coeff1} 
	&   \curl \alpha'   = - \frac12 \Abs{\psi_0}^2 + \frac12 \lan|\psi_0|^2\ran, \  a_0' = \frac{1}{2 } |\psio|^2,\  \theta'  = \frac{1}{2 \chi }\lan|\psio|^2\ran.\end{align}
	     \end{lemma}
\begin{proof}
Though in the present case $n=1$, for tracking reasons, we keep $n$ in the formulae. First, assuming \eqref{s-expans2} holds, we substitute these expressions and 
 \begin{align}\label{psi-expan}      
 \psi_s = s\psio + O_{\cH^{2}_{n}}(s^3)\end{align}
 from \eqref{s-expans1}  into	
 system \eqref{CSeqs-stat-resc}  
to obtain in the leading order
	\begin{align}\label{a0-eq}
	\curl^* a_0' =  \operatorname{Im}\lan\psio, \nabla_{a^n} \psio\ran
	\end{align}
(from the second equation),  the third equation in \eqref{s-expans-coeff1} 
 (from the fourth equation and  $b = \chi + O(s^2)$), and the first equation in \eqref{s-expans-coeff1} (from the third and fourth equation).

	Due to \eqref{eq:J.co.exact} ($ \im (\bar\psio \nabla_{a^n}  \psio) = \frac{1}{ 2} \curl^* |\psio |^2$),  
	equation \eqref{a0-eq} reduces to	
	\begin{equation}\label{a0-eq2}	\curl^* a_0' =  \frac{1}{2} \curl^* \Abs{\psio}^2
	\end{equation}	
	Since $a_0' \in L^2_0(\R^2)$ and $\curl^*$ is invertible on this space, we obtain the second equation in \eqref{s-expans-coeff1}. 
Now, define the functions $\psi_1, \al_{1}, a_{0 1}$ and the number $\theta_1$ by the relations
   \begin{equation}\label{s-expans'}
    \begin{cases}
         \psi_s = s\psio + s^3\psi_1,\ 
  \al_{s}=   \al' s^2 +  s^4\al_{1},\\
     a_{0 s} = a_{0}'  s^2 +   s^4 a_{0 1},\  
         \theta_s = \theta'  s^2 +  s^4\theta_1,             \end{cases}
    \end{equation}
    where $\al', a_{0}'$ and $\theta'$ satisfy
 \eqref{s-expans-coeff1}, plug these expressions into system \eqref{CSeqs-stat-resc} to obtain a system of equations for $\psi_1, \al_{1}, a_{0 1}$ and $\theta_1$ and estimate the latter functions using this system and the implicit function theorem to obtain \eqref{s-expans2}. \end{proof}

	\begin{proof}[Proof of Proposition \ref{prop:lam-b-expan-alt}]	In this proof we omit the subindex $s$. 
	We consider equation \eqref{CSeqs-stat-resc1}, 
	which we multiply scalarly by $\psio$ to obtain
 \begin{align}\label{lam'-comp1-2}
\lan\psio,   - &  \Delta_{ a} \psi + v_\lam (|\psi|^2)   \psi-  a_0 \psi\ran=0.
\end{align}

Rescaling  relation \eqref{barp*-barp}  of Appendix \ref{sec:self-dual} and using the definition $\p_{a}  := \frac12 ((\nabla_{a})_1 - i(\nabla_{a })_2)$ (cf. \eqref{partA} and Remark \ref{rem:holom} of Appendix \ref{sec:self-dual}), we find  (cf. e.g. \cite{GS})
         \begin{align}
\label{barp*-barp-resc}     &-\Delta_{a} =4  \p_{a}^* \p_{a} + \curl a.
        \end{align} 
 Eq \eqref{barp*-barp-resc} yields $-\lan\psio,   \Delta_{ a} \psi\ran=4\lan\p_a\psio,   \p_a \psi\ran+\lan\psio,   \curl a \psi\ran$. Next, the relations  $a=a^{n}+O(s^2)$  and $\psi =s \psio +O(s^3)$ (see \eqref{s-expans1}) and the relation
 \begin{align}\label{dpsio=0}\p_{a^{n}}\psio=0, \end{align}
which follows from \eqref{psi0-eq} and \eqref{barp*-barp-resc}, give $\lan \p_a\psio,  \p_a \psi\ran=O(s^5)$. This implies $-\lan\bar\psio   \Delta_{ a} \psi\ran= \lan\curl a\psio   \psi\ran+O(s^5)$. Now,  constraint \eqref{CSeqs-stat-resc'3} ($\curl a= \lam b-\frac12 |\psi|^2 $) yields $\lan\psio,   \curl a \psi\ran= \lam b\lan\psio,   \psi\ran- \frac12  \lan\psio, |\psi|^2 \psi\ran$, which together with the previous relation and  $\psi =s \psio +O(s^3)$ implies
\begin{align}\label{Del-a-expan-2}-\lan\bar\psio   \Delta_{ a} \psi\ran= \lam b\lan\bar\psio   \psi\ran- \frac12 s^3 \lan|\psio|^4\ran+O(s^5).\end{align}

	Next, using the expansions $v_\lam(\Abs{\psi}^2)  :=\lam v(\frac{1}{\lam} \Abs{\psi}^2)  =\lam v(0) + v'(0) \Abs{\psi}^2	 + O(s^4)$ and  $\psi =s \psio +O(s^3)$, together with the notation $v(0)=V'(0)=-\chi$ and $v'(0)=V''(0)=g$, gives  	
	\begin{align}\label{v-expan-2}	
	\lam v(\frac{1}{\lam} |\psi|^2) & =-\lam \chi + s^2 g |\psio|^2	 + O(s^4).
	\end{align}	
Inserting the last two relations into \eqref{lam'-comp1-2} and using the expansions  $\psi =s \psio +O(s^3)$ and $a_0(s^2) =a'_0 s^2+O(s^4)$, we obtain 
\begin{align}\label{lam'-comp2-2}
 & \lam (b-  \chi)\lan\bar\psio   \psi\ran +  s^3(g-\frac12)\lan|\psio|^4\ran 
       -s^3 \lan a_0' |\phio|^2\ran + O(s^5) =0.
\end{align}	
By \eqref{s-expans1} and $\lam=\frac nb +\theta$ (see \eqref{a-lam-deco}), we have $\lam  =\frac{n}{\chi} + O(s^2)$. Using this and $b=\chi + b'  s^2+O(s^4)$,
 together with  \eqref{lam'-comp2-2} and the 2nd equation in \eqref{s-expans-coeff1} ($a_0' = \frac{1}{2 } \Abs{\psi_0}^2$) and Eq \eqref{psi-expan} ($ \psi_s = s\psio + O_{\cH^{2}_{n}}(s^3)$) give
\begin{align}\label{lam'-comp3-2}
  s^3\big[ \frac n\chi b' \lan|\psio|^2\ran +   (g-\frac12)\lan|\psio|^4\ran & -  \frac{1}{2 } \lan  |\phio|^4\ran\big] + O(s^5) =0.
\end{align}	
Resolving this equation for $ b' $ and using  definition \eqref{beta-tau} and   normalization $\lan|\psio|^2\ran =1$  gives equation  \eqref{b'-alt}. 

To compute $\lam' $, 
we use the definition $\lam = \frac nb+\theta$ to obtain $\lam' =-\frac{n}{\chi^2} b'+\theta'$, which together with  \eqref{b'-alt} and  the third equation in \eqref{s-expans-coeff1} ($\theta'  = \frac{1}{2 \chi }\avg{\Abs{\psi_0}^2}$) and   normalization $\lan|\psio|^2\ran =1$, gives  \eqref{lam'-alt}. 
The remainders are estimated as in the proof of Lemma \ref{prop:ZHK.asymp}. This completes the proof of Proposition \ref{prop:lam-b-expan-alt}.
\end{proof}
Proposition \ref{prop:lam-b-expan-alt} implies 
\begin{corollary}\label{cor:ub} 
For $g\ne 1$ and $s$ 
 sufficiently small, (a) the bifurcating solution exists iff 
\begin{align}\label{chi-g-b-rel}\sign(\chi-b)=-\sign b' = \sign(g-1);\end{align} 
(b) 
the equation $b=b_s$ has a unique solution, $s=s(b)$, for $s$; (c) 
the bifurcation parameter, $s$ and the external magnetic field strength (magnetic flux) $b$ are related as 
\begin{equation} \label{eq:s2mu}
s^2 =- \frac{g-1}{ b'} \mu + O(\mu^2)=\frac n\chi \beta(\tau) 
 \mu+ O(\mu^2),
\end{equation}
where, recall, $\mu:=(\chi-b)/( g-1)$ and $ b'$ is given in \eqref{b'-alt}. 
Consequently, we can express 
  the bifurcation branch, $u_s \equiv u_s^{\hat\lat}$,  given by Theorem \ref{thm:bif-thm-Kdim1},  in terms of $b$ as  \[u^{ b}\equiv u^{\hat\lat, b}:= u^{\hat\lat}_{s(b)}\] and the condition for its existence as  $\mu$ being positive and
 sufficiently small. 
   \end{corollary}
  This  completes the proof of Theorem \ref{thm:ALexist}, apart from the last statement (regarding the energy) which is proven in the next section. 
 
 
 
\DETAILS{According to the first equation in \eqref{s-expans1} ($ b_s =  \chi + O (s^2)$), we can write 
 \begin{align} \label{btilde} 
 b_s =  \chi - \tilde b (s^2), \end{align}  
  where $ \tilde b (s^2)= O (s^2)$. It is shown below 
  that under the condition \eqref{V-cond}, 
  we have $\tilde b'(0) \neq 0$ and therefore, by the inverse function theorem, we can solve the first equation in \eqref{s-expans1} for $s^2$ in terms of $\mu :=  \chi-b$ to obtain}

\section{Energy asymptotic. Proof of Theorem \ref{thm:eng-asymp}}\label{sec:energy-asymp}
Consider  energy functional \eqref{en-resc} defined on  $\latn$-equivariant states $(\psi, a, a_0)$ constrained by \eqref{CSeqs-stat-resc'3} and having the magnetic flux quantization property \eqref{flux-quant-resc}. 
We use the rescaling of equation \eqref{en-repr} proven in Appendix \ref{sec:self-dual} given by 
 \begin{align}\label{en-repr-resc}
\cE_\lam (\psi,  a) =V(0) +\frac{1}{ \lam^2}\lan   |\p_a \psi |^2 & + \frac12(g-1) |\psi |^4+O(|\psi|^6)\ran \notag\\
&\qquad \qquad \qquad   + 2 (\chi - b) (\frac{1}{ \lam}  n- b ). 
\end{align}
where, recall, $\lan f\ran:=\frac{1}{2 \pi}\int_{\Omn} f$ and  $\p_{a}  := \frac12 ((\nabla_{a})_1 - i(\nabla_{a })_2)$. 
 This equation shows that for $V(t)$ a polynomial of the second order, $g=1$ a self-dual regime, see Appendix \ref{sec:self-dual}.

\DETAILS{Now, we use  \eqref{en-repr} to derive 
 the following 
\begin{lemma} \label{lem:ener-asymp} 
Let $(b_s, 
 u_s)$ be the branch of solutions of system \eqref{CSeqs-stat-resc} 
  given in Theorem \ref{thm:bif-thm-Kdim1}. 
Then  
 we have the expansion 
 \begin{align}\label{energy-resc'}
\cE_{\lam_s} (\psi_s,  a_s)& =V(0) +\frac{1}{4 \pi}\frac{ \chi^2}{ n^2} (g-1) s^4\int_{\Omn}  |\psio |^4  
+  2(\frac{n}{ \lam_s} - b_s) (\chi - b_s) +O(s^6),
\end{align}
\end{lemma}}
Let $(b_s, u_s)$ be the branch of solutions of system \eqref{CSeqs-stat-resc} 
  given in Theorem \ref{thm:bif-thm-Kdim1}, $u_s=(\psi_s, \al_{s}, a_{0 s}, \theta_s)$, and let $a_s=a^{n}+  \al_s$ and $\lam_s=\frac{n}{b_s}+  \theta_s$.   We consider the first term on the r.h.s. of  \eqref{en-repr-resc}.  By \eqref{s-expans1},    we have
 \begin{align}\label{psi-expan'}      
 \psi_s = s\psio + O_{\cH^{2}_{n}}(s^3).\end{align}
This equation and Eqs \eqref{dpsio=0} 
($\p_{ a^{n}} \psio=0$), $a_s=a^{n}+  \al_s$ and 
$ \al_s = O_{\cH^{2}_{n}}(s^2)$ imply that   $\p_a \psi_s = O_{\cH^{2}_{n}}(s^3)$. This implies
 \begin{align}\label{est1}
\lan  |\p_a \psi_s |^2\ran =O(s^6). \end{align}
For the second term on the r.h.s. of  \eqref{en-repr-resc}, \eqref{psi-expan'} and  $s^2 = - \frac{g-1}{ b'} \mu + O(\mu^2)$ (see \eqref{eq:s2mu}) give 
\begin{align} 
\lan  |\psi_s |^4\ran =& s^4\lan   |\psio |^4\ran +O(s^6) 
 = (\mu(g-1)/ b')^2\lan   |\psio |^4\ran +O(\mu^3).\notag 
\end{align}
This, together with \eqref{b'-alt} and $\lan|\psi_0|^4 \ran 
=\beta(\tau)$  (remember normalization \eqref{normal-psio} ($\lan|\psio|^2\ran =1$)), yields 
\begin{align}\label{est2}
\lan  |\psi_s |^4\ran =&
 \frac{n^2\mu^2}{\chi^2 \beta(\tau)^2}\beta(\tau)+O(\mu^3).
\end{align}

Next, we evaluate $R:=2(\frac{n}{ \lam_s} - b_s) (\chi - b_s)$. By Proposition \ref{prop:lam-b-expan-alt},  $\frac{n}{ \lam_s}=\chi - \frac{\chi^2}{ n}\lam' s^2+O(s^4)$. Since $\chi - b_s=O(s^2)$, this gives $R=2(\chi - b_s- \frac{\chi^2}{ n}\lam' s^2) (\chi - b_s)+O(s^6)$. Now, due to \eqref{eq:s2mu}, we have $R=2(1 + \frac{\chi^2}{ n}\lam' \frac{1}{ b'})( g-1)^2 \mu^2+O(s^6)$, where, recall, $\mu:=(\chi-b_s)/( g-1)$. Next, by  the definition $\lam_s = \frac{n}{b_s}+\theta$, we have $\lambda' =-\frac{n}{\chi^2} b'  +\theta'$. The last  equation, together with \eqref{b-exp} and \eqref{b'-alt} ($b_s =\chi + b's^2+O(s^4)$ and $ b'=-\frac \chi n (g-1)\beta(\tau)$) and the third equation in \eqref{s-expans-coeff1} ($\theta'  = \frac{1}{2 \chi }\lan|\psi_0|^2\ran= \frac{1}{2 \chi }$), give 
$ 1+\frac{\chi^2}{ n}\lam' \frac{1}{ b'}=\frac{\chi^2}{ n}\frac{\theta'}{ b'} =-\frac{1}{2} \frac{1}{(g-1)\beta(\tau)}.$ 
\DETAILS{using 
 \eqref{lam'-alt} and \eqref{b'-alt}	 gives
\[1-\frac{\chi^2}{ n}\lam' \frac{1}{\tilde b'}=1-\frac{g\beta(\tau) + \frac{1}{2} }{g\beta(\tau)} =-\frac{1}{2} \frac{1}{g\beta(\tau)},\]
which,}
This implies 
 \begin{align}\label{R} 	&2(\frac{n}{ \lam_s} - b_s) (\chi - b_s) 
  = -  
   \frac{g-1}{ \beta(\tau)} \mu^2 +O(\mu^3). 
	\end{align}

We plug \eqref{est1},  \eqref{est2} and \eqref{R} into  \eqref{en-repr-resc} and use the expansion 
 $\lam_s=\frac{n}{\chi}+O(s^2)$ 
   and the relation $s^2 = - \frac{g-1}{ b'} \mu + O(\mu^2)$ (see \eqref{eq:s2mu}) (together with the computation $\frac12(g-1)\frac{1}{ (g-1)^2\beta(\tau)} - \frac{1}{ (g-1)\beta(\tau)}=- \frac{1}{2 (g-1)\beta(\tau)}$) to find 
   \begin{align}\label{cE-asymp}\cE_{\lam_s} (\psi_s,  a_s)	= V(0) - \frac12 \frac{g-1}{\beta(\hat\lat)}\mu^2   + O(\mu^3).\end{align} 
 By \eqref{en-resc'}, $\cE_{\lam_s} (\psi_s,  a_s)= E_{b_s}(\hat\lat)$. The last two relations give   \eqref{Eb-lat-asymp}. This proves Theorem \ref{thm:eng-asymp}. 
\qquad \qquad \qquad \qquad \qquad \qquad \qquad \qquad \qquad \qquad \qquad \qquad \qquad \qquad \qquad  \qquad $\Box$

\section{Energy minimizers} \label{sec:eng.min}

In this section we will prove 
Theorem \ref{thm:ALexist}(c) and Theorem \ref{thm:eng.min}, concerning 
the energy per lattice cell (cf. the proof in \cite{TS2} for the Ginzburg-Landau energy).

The main step in the proof is to express 
 the average energy of the solution $u_s$ in terms 
  of the Abrikosov function \eqref{beta-tau}. 


Let $(b,  u^b)$ be the branch of solutions given by Corollary \ref{cor:ub}. 
 Emphasizing the dependence of the energy on $\tau$ and $b$, we denote  
 \[E_{b}(\tau) := E_{b}(\hat\lat), 
 \] 
where $E_{b}(\hat\lat)$ is given in \eqref{Eb-lat}. 
  We have the following result. 

\begin{theorem}[Energy minimizers] \label{thm:eng.min1}
	Assume  \eqref{V-cond} and \eqref{Aext-cond}, 
	 and let $\mu:=(\chi-b)/( g-1) >0$ be sufficiently small and $g>1$. 
	  Then $\tau_*$ is a non-degenerate saddle point/(local) minimizer of the Abrikosov function $\beta(\tau)$ iff  $\tau_b = \tau_* + O(\mu)$ is a non-degenerate saddle point/(local) minimizer of the energy $E_{b}(\tau)$. \end{theorem}

\begin{proof}[Proof of Theorem \ref{thm:eng.min1}] 
We will derive Theorem \ref{thm:eng.min1} from 
 Theorem \ref{thm:eng-asymp} by using \eqref{en-resc'}. 
By  the energy expansion \eqref{Eb-lat-asymp},
\DETAILS{ is of the form 
	\begin{equation}
	E_{b}(\tau) = V(0) - e_2(\tau) \mu^2 + O(\mu^3)
\end{equation}	
 where $e_2(\tau)$ is independent of $\mu$, $\mu:=\chi-b$, and with $O(\mu^3)$ differentiable in $\tau$ with the same bound $O(\mu^3)$.
Here $\tau$  takes values in the Poincar\'e half-plane $\bH$.}	
  if $\tau_*$ is a non-degenerate critical point of the Abrikosov function, $\beta(\tau)$, 
then, for $\mu$ sufficiently small, $E_{b}(\tau)$ has a unique non-degenerate critical point, $\tau_b$, 
 in a $O(\mu)-$neighbourhood of $\tau_*$.  The converse is also true.
Moreover, 
  $\tau_*$ is a (local) minimum of $\beta(\tau)$ iff  $\tau_b$ is a (local) minimum of $E_{b}(\tau)$. \end{proof}

\begin{proof}[Proof of Theorem \ref{thm:eng.min}] It is shown in \cite{ABN} 
that the critical points of $\beta(\tau)$ are $\tau = e^{\frac{i \pi}{3}}$ and $\tau = e^{\frac{i \pi}{2}}$, with the former a minimum and the latter a maximum. This fact and Theorem \ref{thm:eng.min1} imply Theorem \ref{thm:eng.min}. \end{proof}


\appendix

\section{Nonrelativistic Chern-Simons models} \label{sec:CSmodels} 
In this appendix we provide, for the reader's convenience, some general 
background for non-relativistic Chern-Simons theories.

The Chern-Simons invariant provides a gauge theory in $2+1$ dimension (in addition to the Maxwell/Yang-Mills theories in $3+1$ dimension).  Generally speaking, a gauge theory is defined by an action on connections on a principal bundle, invariant under the corresponding gauge transformations. It is coupled to a matter field (viewed as a section of the associated line/vector bundle) through the covariant derivatives, induced by these connections, acting on this field.

 The Chern-Simons gauge theories on a $2+1$ Minkowski  space $M_3$ involve a gauge field $\a:M_3\ra \R\times \R^2$ and the action functional (
 the CS action)\footnote{For typographical reasons, we use the notation in this appendix different from those in the main text.}
  \begin{align}\label{CS-action}
S_{\rm CS}(\bfa):= 2\int_{M_3} \a \wedge d \a,\ 
\qquad (\text{or }\ S_{\rm CS}(\a):= 2\int_{M_3} \eps^{\mu \nu \rho} \a_\mu \p_\nu \a_\rho). 
\end{align}
It is gauge invariant (under transformations $\a\ra \a+df$ - in the abelian case - satisfying 
$\int_{M_3} d f =0$) and independent of the metric (of a fixed signature). 
Since $S_{\rm CS}(a)$ is independent of the metric, the {\it energy momentum tensor} for it (which is the variation of $S_{\rm CS}(a)$ w.r.to the metric) is zero.

The Euler-Lagrange equation for the Chern-Simons action without external sources, i.e. 
in the vacuum, is $d\a=0$ (in the Minkowski metric), i.e. $\a$ is a flat connection. In a simply connected $M_3$, $\a$ is a pure gauge and the Chern-Simons equation $d\a=0$ has no plane waves like in the EM case. 
\DETAILS{\footnote{The Euler-Lagrange equation for  the Chern-Simons action coupled to  a conserved external current, $J$, i.e., $S_{\rm CS}(\a)-\int  \a\wedge *J$, is
  \begin{align}\notag 
 d a =  *J\ \qquad (\text{or }\ F_{\mu \nu} =\eps_{\mu \nu \rho}   J^\rho).
\end{align}}}

  

The action functional for  the Chern-Simons gauge theory coupled to matter is obtained by using the principle of minimal coupling
\begin{align}\label{CS-matter-action}
S(\phi, a):=S_{\rm matter}(\phi, \a)+ S_{\rm CS}(\a),\ 
\end{align}
where $\phi$ is a matter field or an order parameter (a section of a 
 line, or vector bundle) and $\a$ enters $S_{\rm matter}(\phi, \a)$ through the the 
covariant derivative, $\n_\a$, given locally as $\n_\a:= d - i \a$. 
The Euler-Lagrange equations are
\begin{align}\label{EL-eqs-CS-M1}
&d_\phi S_{\rm matter}(\phi, \a)=0,\\ 
\label{EL-eqs-CS-M2}  & d \a =  *\J,\ \quad \text{with }\ \J:=d_a S_{\rm matter}(\phi, \a),\end{align}
where $d_\chi$ is the G\^ateaux derivative w.r.to $\chi$. 

Apply $d$ to the equation in \eqref{EL-eqs-CS-M2} to obtain  the conservation law, $d* J=0$, or $\p^\mu J_\mu=0$, for the current $J$.  Writing $\a=(a_0,  a)$ and $\J=(\rho, J)$, and introducing the notation,  $B:=(d \a)_0 \equiv\curl  a $ and $E:= (d \a)_{\rm space} \equiv - \n a_0-\p_t  a$, for the magnetic and electric fields, the second CS equation, \eqref{EL-eqs-CS-M2}, can be rewritten as
 \begin{align}\label{CS-eqs-BE}
 E =* J\  
 \quad  \text{ and }\  \quad   B=\rho,
\end{align}
and the conservation law, $\p^\mu \J_\mu=0$, for the current $J$ reads 
 \begin{equation}\label{conserv-law}\p_t \rho+\divv J=0.\end{equation}

For the {\it non-relativistic} model \eqref{NR-CS-action}, 
 the current $\J:=d_\a S_{\rm matter}(\phi, \a)$ has the form 
 \[\J:=\big(|\phi|^2,\ \im (\bar \phi \n_{ a+ a^{\rm ext}} \phi)\big),\] 
 which, together with \eqref{CS-eqs-BE}, gives 
 \eqref{NR-CS-eq2} - \eqref{NR-CS-eq3}. 

 \paragraph{Constraints.}
The second equation in  \eqref{CS-eqs-BE} is a constraint consistent with the evolution. Indeed, taking $\curl$ of the first equation in  \eqref{CS-eqs-BE} and using that $\curl E=-\p_t \curl a$ and $\curl  * J = \divv  J$ gives $\p_t \curl a 
= -\divv  J$. By \eqref{conserv-law}, this in turn implies 
\[\p_t 
(\rho-  B)=\p_t \rho+\divv J=0.\]

On the other hand, the first equation in  \eqref{CS-eqs-BE} follows from the second one and the conservation law, $\p^\mu J_\mu=0$, for the current $J$. 
 Indeed, differentiating the second equation w.r.to $t$  and using \eqref{conserv-law} gives $0= \p_t \curl  a-\p_t \rho= \curl \p_t  a -\divv J$. Furthermore, writing $\divv J=\curl *J$, we find $\curl (\p_t  a- *J)=0$. Therefore there is a function $a_0$ s.t. $\p_t  a- *J= \n a_0$.


 \paragraph{Symmetries.} 
Eqs \eqref{EL-eqs-CS-M1} - \eqref{EL-eqs-CS-M2} (specifically, \eqref{CSeqs}) are invariant under gauge, translation and rotation transformations and Galilean boost (the Lorentz transformations, in the relativistic case). For instance,   the gauge transformation is give by
\begin{equation}\label{g-transf-t-dep}
     T^{\rm gauge}_\chi:\  (\psi,\   a) \mapsto ( e^{i\chi}\psi,\ a + d\chi);
\end{equation}
\DETAILS{ symmetries: 

{\it Gauge symmetry}:  for any sufficiently regular function $\chi : \R^2 \to \R$,
\begin{equation}\label{g-transf}
  T^{\rm gauge}_\chi:\  (\psi(x) ,\    a(x)) \mapsto ( e^{i\chi(x)}\psi(x),\  a(x) + \n\chi(x));
\end{equation}

{\it Translation symmetry}: for any $h \in \R^2$,
\begin{equation}\label{tr-transf}
 T_h^{\rm transl}:\   (\psi(x),\  a(x))  \mapsto (\psi(x + h),\  a(x + h));
\end{equation}

{\it Rotation and reflection symmetry}: for any $R \in O(2)$ (including  the reflections

$f(x)\ra f(-x)$)
\begin{equation}\label{rot-transf}
  T_R^{\rm rot}:\    (\psi(x),\  a(x))  \mapsto (\psi(Rx),\ R^{-1} a(Rx)).
\end{equation}

Note that in the time-dependent case the gauge symmetry should be modified as
\begin{equation}\label{g-transf-t-dep}
     T^{\rm gauge}_\chi:\  (\psi,\   a) \mapsto ( e^{i\chi}\psi,\ a + d\chi);
\end{equation}

 2) The Galilean symmetry (the Poincar\'e  symmetry in the relativistic case).\\}

The spatial symmetries lead to the conservation of the energy and momentum. E.g. for the non-relativistic CS model, the energy is given by  \eqref{en} (recall that the energy momentum tensor for  $S_{\rm CS}(a)$ is zero).

\paragraph{Boundary conditions at infinity.} Depending on the nonlinearity $V(|\phi|^2)$, we can have different boundary conditions at infinity: 

$|\phi|\ra 0$ as $x\ra \infty$, for $V(|\phi|^2):= \frac g2 |\phi|^4$, and 

$|\phi|\ra \sqrt \mu$ as $x\ra \infty$, for $V(|\phi|^2):= \frac g2 (|\phi|^2-\mu)^2, \mu>0$, 

 either b.c. condition, for $V(|\phi|^2):= |\phi|^2 (|\phi|^2-\mu)^2$.

\paragraph{Hamiltonian structure.} 
Since our constraint is conserved by the dynamics, we can avoid using the general theory of Hamiltonian systems on spaces with constraints (the Dirac Poisson brackets).

Take for simplicity $a^{\rm ext}_0=0$ and consider the non-relativistic Lagrangian 
 (cf. \eqref{NR-CS-action})
 \begin{align}\label{NR-CS-action}
S(\phi, a', a_0):=\int i \bar \phi \p_{t a_0'} \phi - 
 |\n_{ a+ a^{\rm ext}} \phi |^2 - V(|\phi|^2) +  \a \cdot\curl \a. 
\end{align}
where  $\phi:\R_+\times \R^2\ra \C$ is the order parameter,  $ \a' = (a, a_0):\R_+\times \R^2\ra \R^2\times \R$ is the CS gauge field,  $a^{\rm ext} $ is an external magnetic potential, $\n_{ a}:=\n +i  a$ and $ \p_{t a_0}:=\p_{t} - i a_0$.
Following the standard procedure, we compute the momenta for  $L (\phi, a, a_0)$, we find
\begin{align}\label{NR-CS-momenta}
d_{\dot\phi} L (\phi, a, a_0) = i \bar \phi, \qquad d_{\dot a'} L (\phi, a, a_0) = - *  a \equiv (a_2, -a_1).
\end{align}
Performing the Legendre transform $\int (i \bar \phi \dot\phi - *a\cdot \dot a) -L (\phi, a, a_0)$ leads to the hamiltonian 
\begin{align}\label{ham}H(\phi, \bar \phi, a, -*  a):=\int  \frac{1}{2} |\n_{ a} \phi |^2 + V(|\phi|^2) - a_0(|\phi|^2  +   2\curl  a). 
\end{align}

Introducing the almost complex structure $j(\phi,  a')=(i\phi, -* a)$, we can write the symplectic form as
\begin{align}\label{sympl-form}\om(\xi,  \al; \eta,  \beta):=\re\int \bar\xi i \eta- \int   \al \cdot *\beta=\lan (\xi,  \al), j(\eta,  \beta)\ran.\end{align} 
Using this symplectic form, for any functional $H(u)$, we define the vector field $\n^\om H(u)$ by the equation $\om(\n^\om H(u); v)=d H(u)v$. With this definition and the Hamiltonian $H$, we see that ZHK equations \eqref{NR-CS-eq1} - \eqref{NR-CS-eq3} can be rewritten as
\begin{align}\label{NR-CS-ham-eq}
& \p_{t } (\phi,  a)=\n^\om H (\phi,  a, a_0), 
\end{align}
{\bf where $H(\phi, a, a_0)\equiv H(\phi, \bar \phi, a, -*  a)$,} with constraint \eqref{NR-CS-eq3}, which can be written as $d_{a_0}H (\phi, a, a_0)=0$. 

\section{The self-dual regime 
} \label{sec:self-dual}   
Introduce the complexified covariant derivatives (harmonic oscillator annihilation and creation  operators, see e.g. \cite{S, S2})  
    \begin{equation}\label{partA}
        \p_{A}  := \frac12 ((\nabla_{A})_1 - i(\nabla_{A })_2). 
    \end{equation}
 and $\p_{A}^*=- \bar\p_{A}$ Remembering the definition of $\n_{A}$, we compute
       $ \p_{A}  = \p + i  A_c, $ 
 where $\p:= \frac12 (\partial_{x_1} - i\partial_{x_2})$ and $A_c:= \frac12 (A_1- i A_2)$.
  We have the following 
\begin{proposition}\label{prop:Wei-form}     \begin{align}
\label{barp-comm}    &[\p_{A},  \p_{A}^*] = \frac12 \curl A;\\
\label{barp*-barp}     &-\Delta_{A} =4  \p_{A}^* \p_{A} + \curl A.
        \end{align} 
   \end{proposition}
 \begin{proof} Let $\n_i:=(\nabla_{A})_i$.  Using  the relations $(\n_1+i\n_2)(\n_1-i\n_2)=\Delta_A - i[\n_1, \n_2]$ and $[\n_1, \n_2]=[\p_1, i A_2]+[i A_1, \p_2]=i\curl A$, we obtain $- \Delta_{ a}= 4\p_A^* \p_A+ \curl A$. Furthermore, $4 [\p_{A},  \p_{A}^*]= [\nabla_1 - i\nabla_2, - \nabla_1 - i\nabla_2]=-2 i[\nabla_1, \nabla_2]= 2 \curl A.$
 \end{proof}   
\begin{corollary}\label{cor:LaplAb-spec} (a) $-\Delta_{A^b}\ge b$ and (b) $\Null(-\Delta_{A^b}- b)=\Null(\p_{A^b})$.   \end{corollary} 
\begin{remark}\label{rem:holom}  {\em (i) Proposition \ref{prop:Wei-form} is formalization of standard results in quantum mechanics, where $\p_{A}$ and $  \p_{A}^*$ are called the annihilation and creation operators (see e.g. \cite{GS}); in geometry, \eqref{barp*-barp} is known as the Weizenb\"ock formula.

(ii) Unlike the standard case in geometry, the formulae above are based on the complex `d', rather than `d-bar', derivative (i.e. on the anti-holomorphic, rather than holomorphic, structure). The reason for it is the unusual sign in the covariant derivatives  $\n_{ A}:=\n +i  A$ and $ \p_{t A_0}:=\p_{t} - i A_0$ due to the basic charge (electron charge) negative ($=-1$). Since, as was already mentioned above, the ZHK equations, unlike the Ginzburg-Landau equations, are not invariant under the transformation $(\psi, A) \ra (\bar\psi, -A)$, we have to stick with the chosen sign in $\n_{ A}$ and $ \p_{t A_0}$.})    \end{remark} 

We say that  the ZHK equations are self-dual iff any local energy minimizing solution, $(\Psi, A, A_{0})$, satisfies
 the first order equation $\p_{A} \Psi=0$ and $\Psi$ is not a constant. The latter equation together with \eqref{NR-CS-eq3} forms a system of the first order equations for $(\Psi, A)$. ($A_{0}$ is found from \eqref{NR-CS-eq2}.)

 \begin{proposition}\label{prop:self-dual}  For self-interaction potential \eqref{V-dw},
 the ZHK equations for $\lat$-equivariant fields are in the self-dual  regime iff $g=1$.   \end{proposition}
We begin with a general result about energy functional \eqref{en}:
\begin{proposition}\label{prop:ener-repr} Consider energy functional \eqref{en}  on  $\lat$-equivariant states $(\Psi, A, A_0)$ constrained by \eqref{NR-CS-eq3} and having the magnetic flux quantization property \eqref{flux-quant}. 
 We assume $V$ satisfies \eqref{V-cond} and let $\chi=-V'(0)$ and $g = V''(0)$. Then we have 
\begin{align}\label{en-repr}
E_{\Oml} (\Psi,  A)& =V(0)|\Oml|+\int_{\Oml}  (| \p_A \Psi |^2  + \frac12(g-1) |\Psi |^4+O(|\Psi|^6))\notag\\
&\qquad \qquad \qquad   + 4 \pi (\chi - b)  n +  2 b (b - \chi )|\Oml|.
\end{align}
 \end{proposition}

 \begin{proof} 
We write $\Om\equiv \Oml$ for a fundamental cell of a lattice $\lat$. Substituting the expression $- \Delta_{ A}= \p_A^* \p_A+ \curl A$ into energy functional \eqref{en}, we find
   \begin{align}\label{en1'}
E_\Om (\Psi,  A) =\int_\Om  (| \p_A \Psi |^2 +\curl A |\Psi |^2 +  V(|\Psi|^2)) . 
\end{align}
Now, we assume $(\Psi,  A)$ satisfies constrain \eqref{NR-CS-eq3}  ($\curl A= b - \frac12 |\Psi|^2$) 
and use this constraint and the definition $V(|\Psi|^2) = \frac{g}{2}(|\Psi|^2 - \mu)^2 $ and the notation $\chi:=g \mu$ to obtain 
 \begin{align}\label{en2'}
E_\Om (\Psi,  A) =\int_\Om  (|\p_A \Psi |^2 + \frac12(g-1) |\Psi |^4 +    (b - \chi) |\Psi |^2 +  \frac12 g\mu^2).
\end{align}
Now 
using  \eqref{NR-CS-eq3} again, we conclude
 \begin{align}\label{en3'}
E_\Om (\Psi,  A)& =\int_\Om   ( | \p_A \Psi |^2 +   \frac12(g-1) |\Psi |^4 + (b - \chi) 2(b - \curl A) +  V(0))\\
\label{en4'}
& =\int_\Om  (| \p_A \Psi |^2  + \frac12(g-1) |\Psi |^4 -  2 (b - \chi)  \curl A)\notag\\
&\qquad \qquad \qquad \qquad \qquad \qquad  \qquad \qquad    + (2 b^2 - 2\chi b+ V(0))|\Om|.
\end{align}
The flux quantization property \eqref{flux-quant} ($\frac{1}{2\pi} \int_{\Om} \curl  A = n\in \Z$) then gives \eqref{en-repr}. \end{proof}

\begin{proof}[Proof of Proposition \ref{prop:ener-repr}] 
 \eqref{en-repr}, for $V(|\Psi|^2)$ of the fourth order, implies the inequality
 \begin{align}\label{en-ineq}
E_\Om (\Psi,  A)& \ge V(0)|\Om|+ 4 \pi (\chi - b)  n +  2 b (b - \chi )|\Om|.
\end{align} 
and that non-constant local energy minimizers saturating the latter inequality are possible only if $g=1$ and  
  $ \p_{A} \Psi=0$, as claimed. 
 \end{proof}

 \section{Proof of Lemma \ref{lem:current-repr}} \label{sec:current-repr}

\begin{proof}[Proof of Lemma \ref{lem:current-repr}] 
	It follows from Eq \eqref{dpsio=0}, 
	 that $\psi_{0}$ satisfies the first order equation
\begin{equation} \label{1stordereqnpsi1}
\left((\nabla_{ a^n})_1- i(\nabla_{ a^n})_2\right)\psi_{0} = 0.
\end{equation}
 Multiplying this relation by $\bar{\psi}_{0}$, we obtain $\bar{\psi}_{0}(\nabla_{a^n})_1\psi_{0}-i\bar{\psi}_{0}(\nabla_{ a^n})_2\psi_{0} =0$. Taking imaginary and real parts of this equation gives
\begin{align*}&\im \bar{\psi}_{0}(\nabla_{ a^n})_1\psi_{0}= \re \bar{\psi}_{0}(\nabla_{ a^n})_2\psi_{0} =\frac{1}{ 2}\p_{x_2}|\psi_{0}|^2,\\
&\im\bar{\psi}_{0}(\nabla_{ a^n})_2\psi_{0} =- \re\bar{\psi}_{0}(\nabla_{ a^n})_1\psi_{0}=- \frac{1}{ 2}\p_{x_1}|\psi_{0}|^2,
\end{align*}
which, in turn, gives \eqref{eq:J.co.exact}. 
 \end{proof}



\begin{thebibliography}{9999}


\bibitem{Abr}
    A.~A.~Abrikosov,
    On the magnetic properties of superconductors of the second group.
    \emph{J. Exper. Theoret. Phys. (JETP)} \textbf{32} (1957), 1147--1182.

\bibitem{ABN}
    A.~Aftalion, X.~Blanc, and F.~Nier,
    Lowest Landau level functional and Bargmann spaces for Bose-Einstein condensates.
	\emph{J. Funct. Anal.} 241 (2006), 661--702.





\bibitem{ACSvH} N. Akerblom, G. Cornelissen, G. Stavenga, J. W. van Holten, Non-relativistic Chern-Simons vortices on the torus., J. Math. Phys. 52 (2011), no. 7.

	
	
	



\bibitem{BH1} I.V. Barashenkov, A.O. Harin, Nonrelativistic Cherns-Simons theory for the repulsive Bose gas, Physical review letters 72 (11), 1575, 1994.	

\bibitem{BH2} I.V. Barashenkov, A.O. Harin, Topological excitations in a condensate of nonrelativistic bosons coupled to Maxwell and Chern-Simons fields, Physical Review D 52 (4), 2471, 1995.


	\bibitem{BouardSaut}
	A.~de~Bouard {L. Berge} and J.~C. Saut.
	\newblock {Blowing up time-dependent solutions of the planar, Chern-Simons
		gauged nonlinear Schr\"odinger equation}.
	\newblock {\em Nonlinearity}, 8(2):235--253, mar 1995.
	
	

	
	\bibitem{ByeonHuhJin1}
	Jaeyoung Byeon, Hyungjin Huh, and Jinmyoung Seok.
	\newblock {Standing waves of nonlinear Schrodinger equations with the gauge
		field}.
	\newblock {\em Journal of Functional Analysis}, 263(6):1575--1608, sep 2012.
	
	\bibitem{ByeonHuhJin2}
	Jaeyoung Byeon, Hyungjin Huh, and Jinmyoung Seok.
	\newblock {On standing waves with a vortex point of order N for the nonlinear
		Chern-Simons-Schrodinger equations}.
	\newblock {\em Journal of Differential Equations}, 261(2):1285--1316, jul 2016.
	
	\bibitem{CERS}
	D.~Chouchkov, N.~M. Ercolani, S.~Rayan, and I.~M. Sigal.
	\newblock {Ginzburg-Landau equations on Riemann surfaces of higher genus}.
	\newblock Annales de l'Institut Henri Poincar\'e C, Analyse non lin\'eaire, 2019.

\bibitem{CSS} I. Chenn, I. M. Sigal, P. Smyrnelis, On Abrikosov lattice solutions of the Ginzburg-Landau equations, Mathematical Physics, Analysis and Geometry, 21:7, 2018.

	
	
	\bibitem{Demoulini2007} Sophia Demoulini.
	\newblock {Global existence for a nonlinear Schroedinger-Chern-Simons system on
		a surface}.
	\newblock {\em Annales de l'Institut Henri Poincare (C) Non Linear Analysis},
	24(2):207--225, mar 2007.
	
	



	
	
	
	
	
	\bibitem{EzawaHoIw}
	Z.~F. Ezawa, M.~Hotta, and A.~Iwazaki.
	\newblock {Nonrelativistic Chern-Simons vortex solitons in external magnetic
		field}.
	\newblock {\em Physical Review D}, 44(2):452--463, jul 1991.
	

\bibitem{Dun} Gerald V. Dunne, Aspects of Chern-Simons Theory, in 
 Proceedings of the 1998 Les Houches Summer School ``Topological Aspects of Low Dimensional Systems''' (Springer, 2000), A. Comtet et al (Eds.), http://arxiv.org/abs/hep-th/9902115v1.


\bibitem{FrEtAl}J. Fr\"ohlich, A. H. Chamseddine, F. Gabbiani, T. Kerler, C. King, P. A. Marchetti, U. M. Studer, E. Thiran, The fractional quantum Hall effect, Chern-Simons theory, and integral lattices, Proceedings of the International Congress
of Mathematicians, Z\"urich, Switzerland 1994 © Birkh\"auser Verlag, Basel, Switzerland 1995.


\bibitem{Girv} Steven M. Girvin, Introduction to the Fractional Quantum Hall Effect, S\'eminaire Poincar\'e 2 (2004) 53--4.




\bibitem{GuoLi} Boling Guo and Fangfang Li, Existence of topological vortices in an Abelian Chern-Simons model, 
J. Math. Physics 56, 101505 (2015); doi: 10.1063/1.4933222.


\bibitem{GS} S. J. Gustafson, I. M. Sigal, Mathematical Concepts of Quantum Mechanics, Second Edition. Universitext, Springer-Verlag Berlin Heidelberg 2011.
	
	
	
	
	

\bibitem{HorvZh} Peter A. Horvathy, Pengming Zhang, Vortices in (Abelian) Chern--Simons gauge theory, Physics Reports 481 (2009) 83--142.

\bibitem{JP} 
	R.~Jackiw and So-Young Pi.
	\newblock {Soliton solutions to the gauged nonlinear Schrodinger equation on
		the plane}.
	\newblock {\em Physical Review Letters}, 64(25):2969--2972, jun 1990.



\bibitem{KivLeeZhang} S. Kivelson, D.-H. Lee and S.-C. Zhang, Scientific American, March, 1996, p. 86.

\bibitem{Loz} G. Lozano, Ground state energy for nonrelativistic bosons coupled to Chern--Simons gauge fields, Phys. Lett. B283 (1992) 70.





\bibitem{Man} N. Manton, First order vortex dynamics, Ann. Phys. (NY) 256 (1997) 114.


 

\bibitem{OP} Sung-Jin  Oh, Fabio Pusateri, Decay and scattering for the Chern-Simons-Schrödinger equations. Int. Math. Res. Not. IMRN 2015, no. 24, 13122--13147.



\bibitem{Ol}  P. Olesen,  Soliton condensation in some self-dual Chern-Simons theories, Phys. Lett. B 265, 361 (1991); Erratum, Phys. Lett. B 267, 541 (1991).

\bibitem{Raj}  K. Rajaratnam,  
On stability/instability of Abrikosov lattice solutions of the ZHK Chern-Simons equations


\bibitem{ReadRez} N. Read and E. Rezayi, Phys. Rev. B 54, 16864 (1996.)



	

\bibitem{S}   
  I.~M.~Sigal,   Magnetic vortices, Abrikosov lattices and automorphic functions, in Mathematical and Computational Modelling (With Applications in Natural and Social Sciences, Engineering, and the Arts), A John Wiley $\&$ Sons, Inc.,   2014.

\bibitem{S2} I.~M. Sigal. \newblock {\em {Partial Differential Equations of Quantum Physics}}. \newblock 2015.



\bibitem{Tong} David Tong, The Quantum Hall Effect, TIFR Infosys Lectures, 2016, arXiv:1606.06687v2.


\bibitem{TS1}  T. Tzaneteas and   I.~M.~Sigal,  Abrikosov lattice solutions of the Ginzburg-Landau equations.   \textit{Contemporary Mathematics} \textbf{535}, 195 -- 213, 2011.


	
	\bibitem{TS2}
	T.~Tzaneteas and I.M. Sigal.
	\newblock {On Abrikosov Lattice Solutions of the Ginzburg-Landau Equation}.
	\newblock {\em Mathematical Modelling of Natural Phenomena}, 8(5):190--205, Sept.
	2013.
	
	
	\bibitem{ZHK}
	S.~C. Zhang, T.~H. Hansson, and S.~Kivelson.
	\newblock {Effective-Field-Theory Model for the Fractional Quantum Hall
		Effect}.
	\newblock {\em Phys. Rev. Lett.}, 62(1):82--85, jan 1989.



\bibitem{Zhang}	 Shou Cheng Zhang, Int. J. Modern Phys. B, vol 6, No 1 (1992) 25-58.
\end{thebibliography}
\end{document}